\documentclass[11pt,reqno,letterpaper]{amsart}

\usepackage{booktabs}
\usepackage{color}
\usepackage{amsmath,amssymb,amsthm}
\usepackage{mathrsfs}
\usepackage{mathtools}
\usepackage{algorithm,algpseudocode}
\usepackage{graphicx}
\usepackage{tikz}
\usepackage[T1]{fontenc}
\usepackage[utf8]{inputenc}
\usepackage[margin=1in]{geometry}
\usepackage[shortlabels]{enumitem}
\usepackage{microtype}
\usepackage{aliascnt}
\usepackage[colorlinks=true,allcolors=blue,pagebackref=true]{hyperref}
\usepackage[noabbrev,capitalize,nameinlink]{cleveref}

\crefname{equation}{}{}
\Crefname{equation}{}{}
\crefname{ALG@line}{Line}{Lines}
\Crefname{ALG@line}{Line}{Lines}

\numberwithin{equation}{section}
\newtheorem{theorem}{Theorem}[section]
\crefname{theorem}{Theorem}{Theorems}
\Crefname{theorem}{Theorem}{Theorems}
\newaliascnt{proposition}{theorem}

\aliascntresetthe{proposition}
\crefname{proposition}{Proposition}{Propositions}
\Crefname{proposition}{Proposition}{Propositions}
\newaliascnt{lemma}{theorem}
\newtheorem{lemma}[lemma]{Lemma}
\aliascntresetthe{lemma}
\crefname{lemma}{Lemma}{Lemmas}
\Crefname{lemma}{Lemma}{Lemmas}
\newaliascnt{observation}{theorem}

\aliascntresetthe{observation}
\crefname{observation}{Observation}{Observations}
\Crefname{observation}{Observation}{Observations}
\newaliascnt{fact}{theorem}

\crefname{fact}{Fact}{Facts}
\Crefname{fact}{Fact}{Facts}
\newaliascnt{claim}{theorem}
\newtheorem{claim}[claim]{Claim}
\aliascntresetthe{claim}
\crefname{claim}{Claim}{Claims}
\Crefname{claim}{Claim}{Claims}
\newaliascnt{corollary}{theorem}

\aliascntresetthe{corollary}
\crefname{corollary}{Corollary}{Corollaries}
\Crefname{corollary}{Corollary}{Corollaries}

\theoremstyle{definition}
\newaliascnt{definition}{theorem}
\newtheorem{definition}[definition]{Definition}
\aliascntresetthe{definition}
\crefname{definition}{Definition}{Definitions}
\Crefname{definition}{Definition}{Definitions}

\theoremstyle{remark}
\newaliascnt{remark}{theorem}
\newtheorem{remark}[remark]{Remark}
\aliascntresetthe{remark}
\crefname{remark}{Remark}{Remarks}
\Crefname{remark}{Remark}{Remarks}

\newcommand{\mb}{\mathbb}

\newcommand{\mc}{\mathcal}

\newcommand{\on}{\operatorname}
\newcommand{\wh}{\widehat}

\newcommand{\R}{\mathbb R}
\newcommand{\Q}{\mathbb Q}
\newcommand{\C}{\mathbb C}
\newcommand{\E}{\mathbb E}
\newcommand{\PP}{\mathbb P}
\newcommand{\one}{\mathbf 1}
\newcommand{\cM}{\mathcal M}

\newcommand{\cD}{\mathcal D}

\newcommand{\norm}[1]{\lVert #1\rVert}

\newcommand{\Var}{\operatorname{Var}}

\algrenewcommand{\algorithmicrequire}{\textbf{Input:}}
\algrenewcommand{\algorithmicensure}{\textbf{Output:}}

\DeclareMathOperator{\per}{per}
\DeclareMathOperator{\diag}{diag}
\DeclareMathOperator{\Cov}{Cov}

\newenvironment{proof*}[1][\proofname]{
  
  \begin{proof}[#1]}{\end{proof}}

\allowdisplaybreaks
\title{A deterministic $(1+\varepsilon)^n$ approximation for the permanent of a nonnegative matrix}

\author{Dingding Dong}
\address{Department of Mathematics, California Institute of Technology, Pasadena, CA 91125, USA}
\email{ddong124@caltech.edu}

\author{Vishesh Jain}
\address{Department of Mathematics, Statistics, and Computer Science,
University of Illinois Chicago, Chicago, IL 60607, USA}
\email{visheshj@uic.edu}

\begin{document}
\begin{abstract}
For every fixed $0<\varepsilon\le1$, we give a deterministic
strongly polynomial algorithm that, given a nonnegative matrix
$A\in\R_{\ge0}^{n\times n}$, returns $Q$ satisfying
$\per A\le Q\le(1+\varepsilon)^n\per A$.
\end{abstract}

\maketitle

\section{Introduction}

Let $A=(A_{ij})\in\R_{\ge0}^{n\times n}$. Its permanent is
\[
 \per A=\sum_{\sigma\in S_n}\prod_{i=1}^n A_{i,\sigma(i)},
\]
where $S_n$ denotes the symmetric group on $n$ elements.
The permanent counts perfect matchings in the bipartite support graph of $A$ when $A$ is a $0$--$1$
matrix. For general nonnegative $A$, it is the total weight of perfect
matchings, where a matching
$M$ has weight $\prod_{ij\in M}A_{ij}$.

Valiant~\cite{Valiant} proved that computing the permanent is
$\#\mathrm P$-complete even for $0$--$1$ matrices. On the algorithmic side, Jerrum, Sinclair, and Vigoda~\cite{JSV} gave
a fully polynomial randomized approximation scheme for
nonnegative matrices. The existence of a deterministic (fully) polynomial
approximation scheme remains a major open problem.

The known deterministic polynomial-time algorithms for general
nonnegative matrices have approximation factors exponential in $n$. Throughout, we measure multiplicative guarantees by the width of a certified interval,~i.e.,~a factor $C \geq 1$ approximation consists of polynomial-time computable
quantities $L(A)$ and $U(A)$ satisfying
\[
  L(A)\leq\on{per}(A)\leq U(A)
  \qquad\text{and}\qquad
  U(A)\leq C L(A).
\]
Linial, Samorodnitsky, and Wigderson~\cite{LSW} obtained a factor
$e^n$ by combining matrix scaling with the van der Waerden inequality.
Subsequent improvements used the Bethe permanent, which can be
computed by convex optimization~\cite{Vontobel}.
Gurvits~\cite{Gurvits} proved that the Bethe permanent is
a lower bound for the permanent, and Gurvits and
Samorodnitsky~\cite{GS} proved the reverse inequality with a
factor $2^n$. Their analysis also yields an approximation
factor of approximately $1.9022^n$, as noted
in~\cite[Footnote~3]{AR}. 
Anari and Rezaei~\cite{AR} proved that the Bethe permanent
approximates the permanent within a factor $2^{n/2}$,
and that this bound is sharp.
Thus a better general approximation requires more than the
Bethe permanent alone.
Recently, Anari~\cite{Anari} and Narang and Perkins~\cite{NP}
obtained approximation factors $(\sqrt{2}-\delta)^n$ for a small absolute constant
$\delta > 0$ by adding corrections to the Bethe approximation.

We show that the base of the exponential factor can be taken to be any constant greater than one.

\begin{theorem}\label{thm:main}
For every fixed $0<\varepsilon\le1$, there is a deterministic
algorithm that, given $A\in\R_{\ge0}^{n\times n}$, returns $Q\in\R_{\ge0}$ satisfying
\begin{equation}\label{eq:main}
 \per A\le Q\le(1+\varepsilon)^n\per A.
\end{equation}
For each fixed $\varepsilon$, the algorithm is strongly polynomial,
using $n^{O(\varepsilon^{-2}\log^3(2/\varepsilon))}$ arithmetic
operations and comparisons.
\end{theorem}

\begin{remark}
We use the computational model of Linial, Samorodnitsky, and
Wigderson~\cite{LSW}, described in \cref{sec:computation}.
The running time can be improved to
$n^{O(\varepsilon^{-1}\log^2(2/\varepsilon))}$
by replacing the map used in \cref{sec:convex-approximation}
with the rational map in \cite[Lemma~2.3]{BarvinokReal}.
We do not pursue this refinement here.
\end{remark}

\subsection{Computational model}\label{sec:computation}

We use the arithmetic model of~\cite{LSW}, in which additions,
subtractions, multiplications, divisions, and exact comparisons
of real numbers have unit cost.
An algorithm is strongly polynomial if its operation count is
polynomial in the number of input entries, independently of
their magnitudes and encoding lengths, and if, on rational
inputs, all numbers computed have bit lengths polynomial in
the total input length.
We apply this definition with $\varepsilon$ fixed.

The bit length of a rational number $r$, denoted as $\ell(r)$, is the sum of the binary
lengths of its numerator and denominator in lowest terms.
Elementary functions are evaluated by rational approximation,
with the cost included in the operation count.

\subsection{Related work}\label{sec:related-work}

Gamarnik and Katz~\cite{GK} obtained a $(1+\varepsilon)^n$
approximation for $0$--$1$ matrices in polynomial time under
the assumption that the support graph has bounded maximum
degree and vertex expansion bounded away from zero.
For arbitrary $0$--$1$ matrices, their algorithm achieves
the same approximation factor in time
$\exp(O(\varepsilon^{-1/2}n^{2/3}\log^3 n))$. Their algorithm uses the approximation scheme of Bayati,
Gamarnik, Katz, Nair, and Tetali~\cite{BGKNT} for the matching
partition function $\sum_M\lambda^{|M|}$, where the sum is over
all matchings. For fixed activity $\lambda$ and bounded maximum
degree, these authors gave a deterministic fully polynomial-time
approximation scheme based on correlation decay.
Gamarnik and Katz use vertex expansion to bound the contribution of
nonperfect matchings when $\lambda$ is sufficiently large.

For matrices satisfying $a\le A_{ij}\le1$ for all $i,j$,
where $a>0$ is fixed, matrix scaling gives a polynomial-factor
approximation in polynomial time~\cite{LSW,BS}.
Barvinok obtained more accurate approximations in
quasipolynomial time, first for real or complex matrices
whose entries are sufficiently close to $1$~\cite{BarvinokComplex},
and subsequently for all real matrices satisfying
$a\le A_{ij}\le1$~\cite{BarvinokPermanents}.
In the latter setting, his algorithm approximates $\log\per A$
to additive error $\zeta\in(0,1)$ in time
$n^{O_a(\log n+\log(1/\zeta))}$. Recently, Yi~\cite{Yi2026} gave an FPTAS for matrices
with entries in $\{0\}\cup[a,1]$, where $a>0$ is fixed,
provided every row and column has at least $(1/2+\gamma)n$ nonzero
entries for a fixed $\gamma>0$. 

Barvinok's algorithms use his interpolation method, which
approximates the logarithm of a polynomial from its first few
coefficients, using bounds on the locations of its zeros.
For the polynomial $Z_G(z)=\sum_M z^{|M|}$, where the sum
is over all matchings of $G$, the Heilmann--Lieb
theorem~\cite{HL} places all zeros on the negative real axis.
Barvinok~\cite{BarvinokReal} used this to approximate the total
number of matchings from the first few coefficients of $Z_G$. Patel and Regts~\cite{PR} showed how to compute the required
coefficients efficiently on bounded-degree graphs, obtaining
polynomial-time approximation algorithms from suitable bounds
on the zeros. We refer to~\cite{Barvinok} for a general account
of this method.

The Bethe permanent is computable by convex
optimization~\cite{Vontobel} and is a lower bound for the
permanent~\cite{Gurvits}.
The authors~\cite{DJ} proved that its approximation factor
improves from $2^{n/2}$ to $2^{2n/g}$ if the support graph has
girth at least $g$, for even $g\ge4$.
This comparison is sharp and holds for arbitrary nonnegative
edge weights.

\medskip

\paragraph{\bf Independent work} After finishing this paper and right before submitting it, we learned of independent work of Kudria, Luo, and Majid~\cite{KLM} which provides an $\exp(o(n))$-factor approximation of the permanent in deterministic polynomial time. We believe that by replacing the exact Taylor polynomial in \cref{sec:convex-approximation} by a suitable approximation, our algorithm can also be modified to achieve this guarantee.  

\subsection{Notation and conventions}\label{sec:notation}

All logarithms are natural, and $0\log0=0$.
An empty product equals one.

For a finite graph $G$, let $\cM(G)$ denote the set of its
matchings, including the empty matching.
For edge weights $w=(w_e)_{e\in E(G)}$, write
\[
 w^M=\prod_{e\in M}w_e,\qquad
 Z_G(w)=\sum_{M\in\cM(G)}w^M.
\]
For nonnegative weights, the associated Gibbs measure is
\[
 \mu_w(M)=\frac{w^M}{Z_G(w)}.
\]
For an induced subgraph $H$ of $G$, $Z_H(w)$ uses the
restrictions of the same weights.
We write $G-v$ for the graph obtained by deleting $v$,
and similarly for several vertices.

For a nonnegative matrix $A$, its bipartite support graph $G_A$
has a left vertex for each row, a right vertex for each column,
and an edge $ij$ whenever $A_{ij}>0$.
We assign this edge weight $A_{ij}$, write $A_e=A_{ij}$ when
$e=ij$, and use $A^M$ with the convention above.

For a finite set $\Omega$, let $\Delta(\Omega)$ denote the set
of probability distributions on $\Omega$.
For $\nu\in\Delta(\Omega)$, its entropy is
\[
 H(\nu)=-\sum_{x\in\Omega}\nu(x)\log\nu(x).
\]
We also write $H(X)$ for the entropy of the distribution of
a random variable $X$, and
\[
 h(t)=-t\log t-(1-t)\log(1-t),\qquad t\in[0,1],
\]
for the binary entropy function.
For $\nu,\mu\in\Delta(\Omega)$, their relative entropy is
\[
 D(\nu\|\mu)=\sum_{x\in\Omega}\nu(x)\log\frac{\nu(x)}{\mu(x)},
\]
with zero summands when $\nu(x)=0$ and value $+\infty$
if $\nu(x)>0=\mu(x)$ for some $x$.

We write $x_+=\max\{x,0\}$ and $\one_{\mathcal A}$ for the
indicator of an event $\mathcal A$.
The vector $\one$ has all entries equal to one, with dimension
determined by context.
Exponentials and logarithms of vectors are taken coordinatewise.
For vectors, $\norm{\cdot}_2$ and $\norm{\cdot}_\infty$ denote
the Euclidean and maximum norms; for matrices,
$\norm{\cdot}_2$ denotes the Euclidean operator norm.
We write $B_2(x,r)$ for the Euclidean ball of radius $r$
centered at $x$.
For real symmetric matrices, $B\preceq C$ means that $C-B$
is positive semidefinite.

\subsection{Proof overview}\label{sec:overview}
Consider the partition function at activity $\lambda > 0$, given by
\[
 Z_G(\lambda A)
 =\sum_{M\in\cM(G)}\lambda^{|M|}A^M
 =\lambda^n\per A
  +\sum_{\substack{M\in\cM(G)\\ |M|<n}}
    \lambda^{|M|}A^M.
\]
For sufficiently large $\lambda$, the contribution of perfect
matchings dominates, so $Z_G(\lambda A)/\lambda^n$ approximates
$\per A$. Gamarnik and Katz~\cite{GK} use this observation to
approximate the permanent of $0$--$1$ matrices. When the support
graph is a bounded-degree expander, they show that the contribution
of nonperfect matchings is sufficiently small for a choice of
$\lambda$ depending only on $\varepsilon$, the expansion constant,
and the degree bound. In particular, $\lambda$ is independent of $n$.
The correlation-decay algorithm of~\cite{BGKNT} approximates
$Z_G(\lambda A)$ in polynomial time for fixed $\lambda$ and
bounded-degree graphs, so this gives a polynomial-time algorithm
for bounded-degree expanders. For general $0$--$1$ matrices,
they use an expander decomposition and obtain a
subexponential-time algorithm. Note that for arbitrary nonnegative weights, the required value of $\lambda$
can be arbitrarily large, even for a graph consisting of two disjoint edges. Indeed, consider $A=\operatorname{diag}(L,L^{-1})$. Then, $\on {per}A = 1$ whereas
\[
 \frac{Z_G(\lambda A)}{\lambda^2}
 =1+\frac{L+L^{-1}}{\lambda}+\frac{1}{\lambda^2},
\]
so that in order for $Z_G(\lambda A)/\lambda^2$ to approximate
$\per A$ within a fixed factor, $\lambda$ must grow at least
proportionally to $L$.

Instead of using a uniform activity $\lambda$, we rescale the rows and columns of $A$
separately. Specifically, we assign an activity $e^{\theta_v}$ to each vertex
and replace the edge weight $A_{uv}$ by
\[
 w_A(\theta)_{uv}=A_{uv}e^{\theta_u+\theta_v}.
\]
Since every perfect matching uses each vertex exactly once, 
\[
 \per w_A(\theta)
 =e^{\sum_v\theta_v}\per A.
\]
On the other hand, the weight of a nonperfect matching is multiplied only by the
activities at the vertices it covers. 

We show that there exists $\theta$ for which
\[
 \per w_A(\theta)
 \le Z_G(w_A(\theta))
 \le e^{\varepsilon n}\per w_A(\theta),
\]
while the total edge weight incident to each vertex is
$O_\varepsilon(1)$. It therefore remains to find such $\theta$
and approximate $\log Z_G(w_A(\theta))$ efficiently. The function $\theta\mapsto\log Z_G(w_A(\theta))$ is convex, but it is unclear how to evaluate it without 
enumerating all
matchings. We overcome this by showing that when the total edge weight at each
vertex is $O_\varepsilon(1)$, the function $\log Z_G(w_A(\theta))$ has an efficiently
computable approximation that remains convex in $\theta$. This allows us to find a suitable $\theta$ by using deterministic convex optimization.

\medskip

\paragraph{\bf Existence of $\theta$} We adapt the approach of
Csikv\'ari~\cite{CsikvariLowerMatching,CsikvariVertexTransitive}. We assume that $\per A>0$ and that every edge in the support of $A$ belongs
to a perfect matching; removing edges violating this can be accomplished in polynomial time and does not change $\on{per}A$. Consider the distribution on matchings corresponding to the weights $w_A(\theta)$,
\[
 \mu_\theta(M) = \mu_{w_A(\theta)}(M)
 =\frac{w_A(\theta)^M}{Z_G(w_A(\theta))}.
\]
The probability of a perfect matching under this distribution
is $\per w_A(\theta)/Z_G(w_A(\theta))$. 
Thus, the approximation we desire is equivalent to requiring that the probability of a matching being perfect is at least $e^{-\varepsilon n}$. In particular, this probability may be exponentially small
in $n$. We will use this flexibility to allow each vertex
to be unmatched with a small probability depending only on
$\varepsilon$, rather than on $n$. 

In Csikv\'ari's vertex-transitive setting~\cite{CsikvariVertexTransitive}, symmetry ensures
that every vertex has the same probability of being matched. Our main idea is to enforce this balance for general weighted graphs by choosing $\theta$ so that
\begin{equation}
\label{eq:matching-prob}
 \PP_{\mu_\theta}(v\text{ is matched})=t
 \qquad\text{for every }v\in V,
\end{equation}
where $t<1$ is chosen close to one. This is somewhat analogous to
matrix scaling~\cite{LSW}, where the row and column factors
are chosen so that the total edge weight incident to each
vertex is one. Note that $t=1$ is impossible for finite $\theta$, since the
empty matching has positive probability under $\mu_\theta$.
Since $\log\per A$ is the maximum of
$H(\nu)+\E_\nu\log A^M$ over distributions on perfect matchings,
it is natural to consider
\[
 F_A(t)=\max_\nu\left\{H(\nu)+\E_\nu\log A^M\right\},
\]
where the maximum is over distributions on matchings which
match every vertex with probability $t$. For $0<t<1$, a standard Lagrange multiplier argument shows
that the maximizing distribution has the form $\mu_\theta$
for some finite $\theta$; see~\cite{SinghVishnoi}.

In \cref{sec:entropy}, we prove the following entropy inequality, which may be of independent interest
\[
 nh(t)\le F_A(t)-t\log\per A\le2nh(t).
\]
Since $h(t) \to 0$ as $t \to 1$, this gives the required approximation by choosing $t<1$
sufficiently close to one in terms of $\varepsilon$.

Finally, as in Csikv\'ari~\cite{CsikvariVertexTransitive}, a
correlation inequality of Heilmann--Lieb~\cite{HL} implies that,
for any $\theta$ satisfying \cref{eq:matching-prob},
\begin{equation}
\label{eq:incident-weight-bound}
 \sum_{e\ni v}w_A(\theta)_e
 \le \frac{t}{(1-t)^2}
 \qquad\text{for every }v\in V.
\end{equation}
Thus, the choice of $t$ above also bounds the total edge
weight at each vertex by $O_\varepsilon(1)$.

\medskip

\paragraph{\bf Convex approximation}
To find $\theta$, we use the dual formulation
\[
 F_A(t)=\min_{\theta\in\R^V}
 \left\{\log Z_G(w_A(\theta))-t\sum_v\theta_v\right\}.
\]
By \eqref{eq:incident-weight-bound}, we can restrict the minimum to
parameters for which the total edge weight at each vertex
is at most $K=O_\varepsilon(1)$. We need an approximation
to $\log Z_G(w_A(\theta))$ that is both efficiently computable
and convex in $\theta$, so that we can minimize the
resulting objective.

Barvinok's interpolation method~\cite{Barvinok,BarvinokReal}
constructs a polynomial approximation $T_m(w)$ to
$\log Z_G(w)$ from the total weights of matchings with at
most $m$ edges. These quantities can be computed by
enumeration, so $T_m$ can be evaluated in $n^{O(m)}$ time. However, the polynomial obtained this way is generally not convex in
$\theta$. We therefore add a correction with sufficiently large Hessian, following a standard
regularization principle; see, for example,
Nesterov~\cite[Theorem~1]{NesterovTensor}.

Concretely, we consider the following regularized version of the Taylor polynomial
\[
 U(w)=T_m(w)+\beta_m\sum_e w_e.
\]
By \eqref{eq:incident-weight-bound}, the regularizer changes $T_m(w)$ by at most $\beta_m n K$ on the region we are optimizing over, whereas a direct computation shows that $\nabla_\theta^2 \sum_e w_e \succeq 0$. The remaining task is therefore to choose $\beta_m$ large enough to compensate for the negative curvature of $T_m$, but small enough to preserve the required accuracy. 

In \cref{sec:convex-approximation}, we prove that
\[
 \nabla_\theta^2 T_m(w_A(\theta))
 \succeq
 -\gamma_m\nabla_\theta^2
 \left(\sum_e w_A(\theta)_e\right),
\]
where $\gamma_m$ depends only on $K$ and $m$ and decreases
exponentially with $m$ for fixed $K$. Taking $\beta_m \geq \gamma_m$ and $m=O(K\log(2K/\eta))$, we get that $U(w_A(\theta))$ is convex and satisfies
\[
 0\le U(w)-\log Z_G(w)\le\eta n
\]
throughout the region we optimize over. 

Finally, to obtain a strongly polynomial algorithm, we use
the preprocessing procedure of Linial, Samorodnitsky, and
Wigderson~\cite{LSW}, followed by rounding, to reduce to a
matrix whose entries have bit length polynomial in $n$ and
$\log(2/\varepsilon)$. This makes the number of operations
required for the convex minimization independent of the
magnitudes of the original entries; see \cref{sec:algorithm}.

\subsection{Acknowledgments} V.J.~is supported by NSF grant DMS-2237646.

\medskip

\paragraph{\bf Statement on AI use} The authors asked GPT 6 Astra to combine the ingredients in the proof of the $(\sqrt{2}-\delta)^n$ approximation \cite{Anari, NP} with their work on sharp girth-dependent bounds for the Bethe permanent~\cite{DJ} in order to obtain a significant improvement over $\sqrt{2}^n$, with a view towards reaching $(1+\varepsilon)^n$. While this attempt was unsuccessful and devolved into rather onerous combinatorial case analysis, it became clear to us after additional interaction that a previous version of \cite{DJ}, which contained a completely different proof, and in particular, a version of~\cref{thm:entropy}, would be useful to overcome these issues. We then supplied Astra with a draft of this old, unpublished version of \cite{DJ}, at which point it was able to prove  \cref{thm:main} without the strongly polynomial guarantee. The authors subsequently strengthened the theorem to achieve this guarantee and substantially simplified and refined various arguments. The authors used Codex for assistance with preparing the manuscript. The mathematical content, the final text, and any errors are the responsibility of the authors.  

\section{An interpolation for the logarithm of the permanent}
\label{sec:entropy}

Let $n\geq1$ and suppose that $A\in\R_{\geq0}^{n\times n}$ satisfies
$\per A>0$. Let $G=(L\sqcup R,E)$ denote the bipartite support graph
of $A$. Deleting support edges that do not belong to any perfect
matching does not change $\per A$, so we henceforth assume that every
$e\in E$ belongs to some perfect matching of $G$.  This property is
used in \cref{eq:edge-coercivity}.

For $0<t\leq1$, let
\[
  \mathcal V_t
  :=
  \left\{
    \nu\in\Delta(\mathcal M(G)):
    \PP_\nu(v\text{ is matched})=t
    \text{ for every }v\in V(G)
  \right\},
\]
and define
\begin{equation*}
  F_A(t)
  :=
  \max_{\nu\in\mathcal V_t}
  \left\{
    H(\nu)+\E_\nu\log A^M
  \right\}.
\end{equation*}
Since $\per A>0$, the graph $G$ has a perfect matching. Retaining
each edge of any fixed perfect matching independently with probability
$t$ gives an element of $\mathcal V_t$, so $\mathcal V_t$ is nonempty.
Moreover, $\mathcal V_t$ is a closed subset of the finite-dimensional
probability simplex $\Delta(\mathcal M(G))$, and hence is compact.
The objective defining $F_A(t)$ is continuous, so the maximum is
attained.

 Let $\pi_A$ denote the probability measure on perfect
matchings given by
\begin{equation*}
  \pi_A(M)=\frac{A^M}{\per A}.
\end{equation*}
For every $\nu\in\mathcal V_1$, since
\[
  D(\nu\|\pi_A)
  =-H(\nu)-\E_\nu\log A^M+\log\per A,
\]
we have
\[
  H(\nu)+\E_\nu\log A^M
  =
  \log\per A-D(\nu\|\pi_A), \qquad \nu \in \mc V_1.
\]
Since $D(\nu\|\pi_A)\geq0$, with equality at $\nu=\pi_A$, it follows that
\[
  F_A(1)=\log\per A.
\]
Thus $F_A(t)$ may be viewed as an interpolation for $\log \per A$ through 
matching distributions with prescribed vertex marginals. The following theorem quantifies this interpolation.

\begin{theorem}\label{thm:entropy}
For every $A\in\R_{\geq0}^{n\times n}$ with $\per A>0$ and every
$0<t<1$,
\begin{equation*}
  t\log\per A+nh(t)
  \leq F_A(t)
  \leq t\log\per A+2nh(t).
\end{equation*}
\end{theorem}

\begin{remark}\label{rem:entropy-sharpness}
Both constants in \cref{thm:entropy} are best possible. The lower
bound is attained for the identity matrix $I_n$, while the constant
$2$ in the upper bound is asymptotically sharp for the all-ones
matrix $J_n$, for example, with $n=r^2$ and $t=1-r^{-1}$ as $r\to\infty$. 
\end{remark}

The lower bound is obtained by thinning a perfect matching sampled
from $\pi_A$.

\begin{proof}[Proof of the lower bound in \cref{thm:entropy}]
Let $\sigma\sim\pi_A$. Then
\[
  H(\sigma)+\E\log A^\sigma=\log\per A.
\]
Retain each edge of $\sigma$ independently with probability $t$, and
let $M$ denote the resulting matching. The law of $M$ belongs to
$\mathcal V_t$, and, by linearity of expectation,
\[
  \E\log A^M=t\,\E\log A^\sigma.
\]

View $\sigma$ as a bijection from
$L=\{v_1,\ldots,v_n\}$ to $R$, and let $S\subseteq L$ be the set of
vertices whose incident edge in $\sigma$ is retained. For
$I\subseteq L$, write $\sigma_I$ for the restriction of $\sigma$ to
$I$. The matching $M$ is in bijection with the pair $(S,\sigma_S)$.
Therefore,
\begin{align*}
  H(M)
  &=H(S)+H(\sigma_S\mid S)\\
  &=H(S)+\sum_{I\subseteq L}
    \PP(S=I)H(\sigma_I\mid S=I)\\
  &=nh(t)+\sum_{I\subseteq L}
    \PP(S=I)H(\sigma_I),
\end{align*}
where the last equality uses the independence of $S$ and $\sigma$.
Equivalently,
\[
  H(M)=nh(t)+\E_S H(\sigma_S).
\]

By the fractional-cover form of Shearer's entropy inequality (cf.~\cite{MadimanTetali}),
\[
  \E_S H(\sigma_S)\geq tH(\sigma).
\]
For completeness, we include the short proof. Suppose $L=\{v_1,\dots,v_n\}$. For every fixed
$I=\{v_{i_1},\dots,v_{i_k}\}\subseteq L$ with $i_1<\cdots<i_k$, the chain rule and the fact that conditioning reduces
entropy give
\[
  H(\sigma_I)=\sum_{j=1}^kH\!\left(
    \sigma(v_{i_j})
    \,\middle|\,
    \sigma(v_{i_1}),\ldots,\sigma(v_{i_{j-1}})
  \right)
  \geq
  \sum_{i=1}^n
  \one_{\{v_i\in I\}}
  H\!\left(
    \sigma(v_i)
    \,\middle|\,
    \sigma(v_1),\ldots,\sigma(v_{i-1})
  \right).
\]
Taking expectations over $S$ and using
$\PP(v_i\in S)=t$ for every $i$ yields
\[
  \E_S H(\sigma_S)\geq tH(\sigma).
\]
Thus,
\[
  F_A(t)
  \geq H(M)+\E\log A^M
  \geq nh(t)+t\bigl(H(\sigma)+\E\log A^\sigma\bigr)
  =nh(t)+t\log\per A.
  \qedhere
\]
\end{proof}

The remainder of this section is devoted to the proof of the upper bound in \cref{thm:entropy}.

\medskip

\subsection{The maximizing distribution} We now express $F_A(t)$ as a convex minimization problem.
This uses standard entropy duality; see Wainwright and
Jordan~\cite{WainwrightJordan} for background and Singh and
Vishnoi~\cite{SinghVishnoi} for its use in counting algorithms.
We give the proof in our setting below.
Fix $0<t<1$. For $\theta\in\R^{L\sqcup R}$, define
\[
  w_A(\theta)_{uv}=A_{uv}e^{\theta_u+\theta_v},
  \qquad uv\in E,
\]
and let $\mu_\theta:=\mu_{w_A(\theta)}$, where recall that $\mu_w(M)=w^M/Z_G(w)$ is the Gibbs measure
associated with the edge weights $w$. 
Define
\begin{equation*}
  \Phi_{A,t}(\theta)
  :=\log Z_G(w_A(\theta))-t\sum_v\theta_v.
\end{equation*}
For every vertex $v$ and matching $M$, let $\chi_v(M)$ denote the indicator that $v$ is matched in $M$. Then we have
\[
\mu_\theta(M)=\frac{A^M\prod_v e^{\theta_v\chi_v(M)}}{Z_G(w_A(\theta))}.
\]
For every $\nu\in\mathcal V_t$,
\begin{align*}
  D(\nu\|\mu_\theta)
  &=
  -H(\nu)-\E_\nu\log\mu_\theta(M)=
  -H(\nu)-\E_\nu\left(\log A^M+\sum_v\theta_v\chi_v(M)-\log Z_G(w_A(\theta))\right)\\
  &=
  -H(\nu)-\E_\nu\log A^M
  -t\sum_v\theta_v
  +\log Z_G(w_A(\theta)),
\end{align*}
where the last equality uses
$\PP_\nu(v\text{ is matched})=t$ for every $v$. Rearranging gives
\begin{equation}\label{eq:gibbs-variational}
  H(\nu)+\E_\nu\log A^M
  =
  \Phi_{A,t}(\theta)-D(\nu\|\mu_\theta)
  \leq
  \Phi_{A,t}(\theta).
\end{equation}
Thus
\[
  F_A(t)\leq\Phi_{A,t}(\theta)
  \qquad\text{for every }\theta\in\R^{L\sqcup R},
\]
with equality whenever $\mu_\theta\in\mathcal V_t$. We next show that $\Phi_{A,t}$ attains its infimum and that a
minimizer $\theta$ satisfies $\mu_\theta\in\mathcal V_t$.
\begin{lemma}\label{lem:balanced-gibbs}
For every $0<t<1$, the infimum of $\Phi_{A,t}$ over
$\theta\in \R^{L\sqcup R}$ is attained. If $\theta_t$ is any minimizer, then
$\mu_{\theta_t}\in\mathcal V_t$ and
\begin{equation}\label{eq:balanced-dual}
  F_A(t)
  =
  \Phi_{A,t}(\theta_t)
  =
  \min_{\theta\in\R^{L\sqcup R}}\Phi_{A,t}(\theta).
\end{equation}
\end{lemma}

\begin{proof}
We first prove that the infimum is attained. 
Fix $\theta\in\R^{L\sqcup R}$, choose an edge $e_0$ maximizing
$|\theta_u+\theta_v|$, and let $P$ be a perfect matching containing
$e_0$. (Note that such a perfect matching exists by the support reduction at the start of the section.) Let
\[
  P_+:=\{uv\in P:\theta_u+\theta_v>0\}, \qquad c_A:=\max_{M\in\cM(G)}|\log A^M|.
\]
Since $P_+$ is a matching,
\[
  \log Z_G(w_A(\theta))
  \geq \log w_A(\theta)^{P_+}=\log A^{P_+}+\sum_{uv\in P_+}(\theta_u+\theta_v)
  \geq -c_A+\sum_{uv\in P}(\theta_u+\theta_v)_+.
\]
Moreover, since $P$ is a perfect matching,
\begin{equation*}
  \sum_v\theta_v=\sum_{uv\in P}(\theta_u+\theta_v).
\end{equation*}
Therefore,
\begin{align}
  \Phi_{A,t}(\theta)
  &=
  \log Z_G(w_A(\theta))
  -t\sum_{uv\in P}(\theta_u+\theta_v) \nonumber\\
  &\geq
  \sum_{uv\in P}(\theta_u+\theta_v)_+
  -t\sum_{uv\in P}(\theta_u+\theta_v)
  -c_A \nonumber\\
  &\geq \min\{t,1-t\}
  \sum_{uv\in P}|\theta_u+\theta_v|
  -c_A\nonumber\\
  &\geq \min\{t,1-t\}
  \max_{uv\in E}|\theta_u+\theta_v|
  -c_A,
  \label{eq:edge-coercivity}
\end{align}
where the last step follows from that $e_0\in P$.
Note from above that $\Phi_{A,t}(\theta)$ depends on $\theta$ only through the
vector
$
  (\theta_u+\theta_v)_{uv\in E}.
$
Thus, letting $S:=\{(\theta_u+\theta_v)_{uv\in E}: \theta\in \mb{R}^{L \sqcup R}\}$, it makes sense to define the map $\Psi:S\to\R$ by
\[
\Psi(b)=\Phi_{A,t}(\theta)\qquad  \text{for any $\theta\in \R^{L\sqcup R}$ with $b=(\theta_u+\theta_v)_{uv\in E}$.}
\]
 Note that $S$ is a linear subspace of $\R^{E}$, and hence is closed.

Since $\Phi_{A,t}$ is continuous, so is $\Psi$.
By \cref{eq:edge-coercivity},
for every $C\in\R$, the set $\{b \in S:\Psi(b)\leq C\}\subseteq S$ is closed and bounded in $S$. Since $S$ itself is closed in $\R^{E}$, we get that $\{b \in S:\Psi(b)\leq C\}$ is compact.

Thus $\Psi$ attains its infimum on the compact nonempty sublevel set
$\{b\in S:\Psi(b)\leq\Psi(0)\}$, so the infimum of $\Phi_{A,t}$ is attained.

Finally, since 
\begin{align}\label{eq:vertex-marginal}
\frac{\partial}{\partial\theta_v}w_A(\theta)^M&=w_A(\theta)^M\chi_v(M),\notag\\
    \frac{\partial}{\partial\theta_v}
  \log Z_G(w_A(\theta))
  &=
  \frac{1}{Z_G(w_A(\theta))}
  \sum_{M\in\cM(G)}
  w_A(\theta)^M
  \chi_v(M)
  =
  \PP_{\mu_\theta}(v\text{ is matched}),
\end{align}
we have
\begin{equation*}
  \frac{\partial\Phi_{A,t}}{\partial\theta_v}
  =
  \PP_{\mu_\theta}(v\text{ is matched})-t.
\end{equation*}
Hence, at any minimizer $\theta_t$,
\[
  \PP_{\mu_{\theta_t}}(v\text{ is matched})=t
  \qquad\text{for every }v,
\]
so $\mu_{\theta_t}\in\mathcal V_t$. Taking
$\nu=\mu_{\theta_t}$ in \cref{eq:gibbs-variational} gives
\[
  F_A(t)=\Phi_{A,t}(\theta_t),
\]
which proves \cref{eq:balanced-dual}.
\end{proof}
\medskip

\subsection{A bound on the partition function}
We next bound the partition function at a minimizer of
$\Phi_{A,t}$.

\begin{lemma}\label{lem:partition-bound}
For every $0<t<1$ and every minimizer $\theta_t$ of
$\Phi_{A,t}$,
\begin{equation}\label{eq:logZ-bound}
  0
  \leq
  \log Z_G(w_A(\theta_t))
  \leq
  -2n\log(1-t).
\end{equation}
\end{lemma}

\begin{proof}
Fix a minimizer $\theta_t$ and for convenience of notation, write
$
  w:=w_A(\theta_t).
$
By \cref{lem:balanced-gibbs}, every vertex is matched with probability
$t$ under $\mu_w$. 

We will need the following correlation inequality of  Heilmann and Lieb~\cite{HL}. For the reader's convenience, we give a proof here. 

\begin{claim}
    If $u\in L$ and $v\in R$, then
\begin{equation}\label{eq:monomer-ineq}
  Z_G(w)Z_{G-u-v}(w)
  \geq
  Z_{G-u}(w)Z_{G-v}(w).
\end{equation}
\end{claim}
\begin{proof*}
    Consider the map
\[
  \cM(G-u)\times\cM(G-v)
  \longrightarrow
  \cM(G)\times\cM(G-u-v)
\]
defined as follows. Given a red matching $R\in\cM(G-u)$ and a blue
matching $B\in\cM(G-v)$, regard any common edge as two parallel
colored copies. 

If $B$ also avoids $u$, leave the pair unchanged.
Otherwise, the component of $R\cup B$ containing $u$ is an
alternating path beginning with $u$ and a blue edge incident to it. We observe that this path cannot contain
$v$. Indeed, since $B$ does not contain $v$, $v$ cannot be an internal vertex of this path. Furthermore, if $v$ is the other endpoint, then the last edge must be red; in particular, the path must have even length. But $u$ and $v$ lie in different parts of the bipartite graph, so the path must have odd length, which is a contradiction. Hence $v$ does not lie on this path. Swap the two colors along this path. After the swap, the blue matching avoids both $u$ and $v$, while the
red edges still form a matching of $G$. 

We now show that the map is injective. In the image, there is a red edge incident to
$u$ if and only if a swap occurred. In this case, the component
containing $u$ is the same alternating path, and swapping its colors
again recovers the original pair. If there is no red edge at $u$, the
map was the identity. The map preserves the multiset of edges, and hence preserves the
product weight $w^R w^B$. Summing over all input pairs gives
\cref{eq:monomer-ineq}.
\end{proof*}

Now let $e=uv$ and note that
\[
  \PP_{\mu_w}(e\in M)
  =\frac{1}{Z_G(w)}\sum_{M'\in\cM(G-u-v)}w_ew^{M'}=
  w_e\frac{Z_{G-u-v}(w)}{Z_G(w)}
  =
  w_e\PP_{\mu_w}(u,v\text{ both unmatched}).
\]
Dividing both sides of \cref{eq:monomer-ineq} by $Z_G(w)^2$, we have
\[
  \PP_{\mu_w}(u,v\text{ both unmatched})
  \geq
  \PP_{\mu_w}(u\text{ unmatched})
  \PP_{\mu_w}(v\text{ unmatched})
  =(1-t)^2.
\]
Hence, for every edge $e$,
\[
  \PP_{\mu_w}(e\in M)\geq w_e(1-t)^2.
\]
Summing over the edges incident to a vertex gives
\begin{equation}\label{eq:degree-bound}
  \sum_{e\ni v}w_e
  \leq
  \frac{\sum_{e \ni v} \mb{P}_{\mu_w}(e \in M)}{(1-t)^2} = \frac{\mb{P}_{\mu_w}(v \text{ is matched})}{(1-t)^2} = \frac{t}{(1-t)^2} .
\end{equation}

Finally, every matching is obtained by letting each vertex in $L$
choose either no edge or one incident edge, subject to the constraint that the chosen edges form a matching. Dropping the matching constraint gives the upper bound
\[
  Z_G(w)
  \leq
  \prod_{v\in L}
  \left(1+\sum_{e\ni v}w_e\right)
  \leq
  \left(1+\frac{t}{(1-t)^2}\right)^n
  \leq
  (1-t)^{-2n}.
\]
The empty matching gives $Z_G(w)\geq1$, completing the proof.
\end{proof}

\begin{remark}
    The bound on the edge weights in~\eqref{eq:degree-bound} is
related to a result of Kahn and Kayll~\cite{KahnKayll}. They showed that if
\[
 \bigl(\PP_{\mu_w}(e\in M)\bigr)_{e\in E}
 \in(1-\sigma)\mathcal P(G),
\]
where $\mathcal P(G)$ is the matching polytope and
$0<\sigma<1$, then each edge weight $w_e$ is bounded
by a constant depending only on $\sigma$.

In our setting, write $x_e=\PP_{\mu_w}(e\in M)$. Since
$\sum_{e\ni v}x_e=t$ for every vertex and $G$ is bipartite,
$x/t$ belongs to the perfect matching polytope. Thus
$x\in t\mathcal P(G)$, and their result applies with
$\sigma=1-t$. The argument above gives the explicit bound
in~\eqref{eq:degree-bound} on the sum of the edge weights
at each vertex, which we will use in constructing the
polynomial approximation.
\end{remark}

\medskip

\subsection{Completing the upper bound}
We now complete the proof of the upper bound. We work with finite differences, avoiding the need to establish
differentiability of $F_A(t)$ or continuity of $\theta_t$.

\begin{proof}[Proof of the upper bound in \cref{thm:entropy}]
Fix $0<t<s<1$, and let $\theta_s$ be a minimizer of
$\Phi_{A,s}$. By \cref{lem:balanced-gibbs},
\[
  F_A(s)
  =\Phi_{A,s}(\theta_s)
  =\log Z_G(w_A(\theta_s))
  -s\sum_v(\theta_s)_v,
\]
whereas the variational bound at $t$ gives
\[
  F_A(t)=\min_{\theta\in\R^{L\sqcup R}}\Phi_{A,t}(\theta)\leq
  \log Z_G(w_A(\theta_s))
  -t\sum_v(\theta_s)_v.
\]
Dividing by $s$ and $t$, respectively, and subtracting gives
\begin{equation}\label{eq:envelope}
  \frac{F_A(t)}{t}-\frac{F_A(s)}{s}
  \leq
  \left(\frac1t-\frac1s\right)
  \log Z_G(w_A(\theta_s)).
\end{equation}

Fix $t<b<1$ and an equipartition
\[
  t=t_0<t_1<\cdots<t_k=b.
\]
Applying \cref{eq:envelope} to each pair
$(t_{j-1},t_j)$ and using \cref{eq:logZ-bound} gives
\[
  \frac{F_A(t)}t-\frac{F_A(b)}b
  \leq
  -2n\sum_{j=1}^k
  \left(\frac1{t_{j-1}}-\frac1{t_j}\right)
  \log(1-t_j).
\]
Letting $k\to\infty$ gives
\begin{equation}\label{eq:integrated-envelope}
  \frac{F_A(t)}t-\frac{F_A(b)}b
  \leq
  -2n\int_t^b\frac{\log(1-u)}{u^2}\,du.
\end{equation}

Next, we show the endpoint bound
\begin{equation}\label{eq:entropy-endpoint}
  \limsup_{s\uparrow1}F_A(s)\leq\log\per A.
\end{equation}
Choose $s_k\uparrow1$ such that
$F_A(s_k)\to\limsup_{s\uparrow1}F_A(s)$, and let
$\nu_k\in\mathcal V_{s_k}$ attain $F_A(s_k)$. Since
$\Delta(\mathcal M(G))$ is compact, after passing to a subsequence,
$\nu_k$ converges to a probability measure $\nu_*$ on $\mathcal M(G)$.
For every vertex $v$, continuity gives
$\PP_{\nu_*}(v\text{ is matched})=\lim_k s_k=1$, so $\nu_*\in\mathcal V_1$.
Since $\mathcal M(G)$ is finite, both $H(\nu)$ and $\E_{\nu}\log A^M$ are continuous functions of $\nu$. Hence
\[
  \lim_{k\to\infty}F_A(s_k)
  =
  \lim_{k\to\infty}
  \left(
    H(\nu_k)+\E_{\nu_k}\log A^M
  \right)
  =
  H(\nu_*)+\E_{\nu_*}\log A^M \leq F_A(1) = \log \per A.
\]
Finally, letting $b\uparrow1$ in \cref{eq:integrated-envelope} and using \cref{eq:entropy-endpoint}, we obtain
\[
  \frac{F_A(t)}t-\log\per A
  \leq
  -2n\int_t^1\frac{\log(1-u)}{u^2}\,du
  =
  \frac{2nh(t)}t.
\]
Multiplying by $t$ proves the upper bound.
\end{proof}

\section{A convex polynomial approximation}\label{sec:convex-approximation}

Let $G=(V,E)$ be a finite simple graph with $N=|V|\ge2$, and let
$w=(w_e)_{e\in E}$ be positive edge weights. We will approximate
$\log Z_G(w)$ by a low-degree polynomial $U(w)$. In the application,
$w_{uv}=A_{uv}e^{\theta_u+\theta_v}$, so the logarithmic edge weights
$y_{uv}=\log A_{uv}+\theta_u+\theta_v$ are affine functions of the
vertex fields. Since we will later optimize $U(e^y)$ over the vertex fields $\theta_v$, we seek an approximation for which
$y\mapsto U(e^y)$ is convex.

Fix a rational $K\ge1$, and define the convex set
\[
 \cD_K=\left\{y\in\R^E:
       \sum_{e\ni v}e^{y_e}\le K\text{ for every }v\in V\right\}.
\]
Throughout this section, derivatives
are taken with respect to the logarithmic weights $y$, whereas
polynomial degrees refer to the weights $w$. In particular, with
$w=e^y$, we write
\[
 \partial_e:=\frac{\partial}{\partial y_e},\qquad
 \partial_e[F(e^y)]=w_e\frac{\partial F}{\partial w_e}(w).
\]
\begin{theorem}\label{thm:surrogate}
For rational $K\ge1$ and $0<\eta\le1$, there is an explicit polynomial
$U(w)$ with rational coefficients and degree at most
\[
 m=O\!\left(K\log\frac{2K}{\eta}\right)
\]
such that $y\mapsto U(e^y)$ is convex on $\cD_K$ and
\[
 0\le U(e^y)-\log Z_G(e^y)\le\eta N/2
 \qquad\text{for every }y\in\cD_K.
\]
The polynomial can be constructed using $N^{O(m)}$ arithmetic
operations. Given $w=e^y$, its value and any fixed-order derivatives
of $y\mapsto U(e^y)$ can also be evaluated using $N^{O(m)}$
arithmetic operations. 
\end{theorem}


\subsection{The polynomial approximation}

We first construct a polynomial $T_m(w)$ approximating $\log Z_G(w)$.
To ensure convexity in the logarithmic weights, we will also bound
the Hessian of the approximation error and add a small multiple of
\[
 W(w)=\sum_{e\in E}w_e.
\]
The same correction will make the approximation an upper bound.

One could try expanding $\log Z_G(sw)$ at $s=0$ and evaluating the
Taylor polynomial at $s=1$. However, zeros of $Z_G(sw)$ may lie too
close to the origin for the Taylor series to converge at $s=1$.
We instead make the change of variables
\begin{equation}\label{eq:interpolation-map}
 \rho=\frac1{4K},\qquad
 \tau(z)=\frac{\rho z}{1-z},\qquad
 z_\star=\frac{4K}{4K+1}.
\end{equation}
Then $\tau(0)=0$ and $\tau(z_\star)=1$. This change of variables
is a rescaled version of Barvinok's transformation
\cite[Lemma~2.4]{BarvinokReal}. We will show below that, for
$y\in\cD_K$, the function $z\mapsto\log Z_G(\tau(z)e^y)$
is analytic on $|z|<1$, with the logarithm chosen to vanish
at $z=0$.

For $m\ge0$, define
\begin{equation*}
 T_m(w)=\sum_{\ell=0}^m z_\star^\ell
       [z^\ell]\bigl(\log Z_G(\tau(z)w)\bigr),
 \qquad
 \mathcal E_m(w)=\log Z_G(w)-T_m(w),
\end{equation*}
where $[z^\ell]f(z)$ denotes the coefficient of $z^\ell$ in the
Taylor expansion of $f$ at zero. Thus $T_m$ is the Taylor polynomial in $z$,
evaluated at the point $z_\star$ corresponding to the original weights.
The next lemma shows that $T_m$ can be computed by enumerating
matchings with at most $m$ edges.

\begin{lemma}\label{lem:taylor-construction}
For $m\ge1$, the function $T_m(w)$ is a polynomial with rational
coefficients and degree at most $m$. It can be constructed using
$N^{O(m)}$ arithmetic operations. Given $w=e^y$, its value and any
fixed-order derivatives of $y\mapsto T_m(e^y)$ can be evaluated using $N^{O(m)}$ arithmetic operations.
\end{lemma}

\begin{proof}
Let
\[
 M_j(w)=\sum_{\substack{M\in\cM(G)\\|M|=j}}w^M
\]
be the total weight of the matchings with $j$ edges. Since $\tau(0)=0$,
the coefficients through degree $m$ in
\[
 Z_G(\tau(z)w)=1+\sum_{j\ge1}M_j(w)\tau(z)^j
\]
depend only on $M_1(w),\ldots,M_m(w)$. Write
\[
 Z_G(\tau(z)w)=1+\sum_{\ell\ge1}s_\ell(w)z^\ell,
 \qquad
 \log Z_G(\tau(z)w)=\sum_{\ell\ge1}\lambda_\ell(w)z^\ell.
\]
Expanding $\tau(z)^j=\rho^jz^j(1-z)^{-j}$ gives
\[
 s_\ell(w)=
 \sum_{j=1}^{\min\{\ell,\lfloor N/2\rfloor\}}
 \rho^j\binom{\ell-1}{j-1}M_j(w).
\]
Differentiating the two series with respect to $z$ and using
$(\log Z)'Z=Z'$ gives the recurrence
\begin{equation}\label{eq:log-coefficient-recurrence}
 \lambda_\ell(w)
 =s_\ell(w)-\frac1\ell
   \sum_{j=1}^{\ell-1}j\lambda_j(w)s_{\ell-j}(w).
\end{equation}
It follows inductively that $\lambda_\ell(w)$ has rational
coefficients and degree at most $\ell$. Hence
$T_m(w)=\sum_{\ell=1}^m z_\star^\ell\lambda_\ell(w)$ has rational
coefficients and degree at most $m$.

The matchings with at most $m$ edges can be enumerated in
$N^{2m+O(1)}\operatorname{poly}(m)$ operations. The recurrence then
constructs $T_m$ in $N^{O(m)}$ operations since all intermediate
polynomials have degree at most $m$ in at most $N^2$ variables,
and therefore have $N^{O(m)}$ monomials. Finally, for a monomial
$w^\gamma=\prod_e w_e^{\gamma_e}$,
\[
 \partial_e w^\gamma=\gamma_e w^\gamma.
\]
Thus evaluation and any fixed-order derivatives can also be
computed in $N^{O(m)}$ operations.
\end{proof}

We next state the error bounds needed to prove \cref{thm:surrogate}.
Let
\begin{equation*}
 q=\frac{8K}{8K+1},\qquad
 \beta_m=\frac{400Kq^{m+1}}{1-q}.
\end{equation*}

\begin{lemma}\label{lem:taylor-error-bounds}
For every $m\ge0$ and $y\in\cD_K$,
\begin{align}
 |\mathcal E_m(e^y)|
 &\le \beta_mW(e^y),\label{eq:taylor-value-error}\\
 \left|a^{\mathsf T}\nabla_y^2\mathcal E_m(e^y)a\right|
 &\le\beta_m\sum_e e^{y_e}a_e^2
 \qquad(a\in\R^E).\notag
\end{align}
Moreover, for every edge $e$,
\begin{equation}\label{eq:taylor-gradient-error}
 |\partial_e\mathcal E_m(e^y)|\le \beta_m e^{y_e}.
\end{equation}
\end{lemma}

We first deduce the theorem from this lemma. The value and Hessian
bounds suffice for this deduction; the gradient bound will also be
used in the numerical minimization.

\begin{proof}[Proof of \cref{thm:surrogate}]
If $E=\varnothing$, take $U=0$. Otherwise, choose
\begin{equation}\label{eq:surrogate-degree}
 m=\min\left\{j\in\mathbb N:
       2\beta_jK\le\eta\right\},
\end{equation}
and set
\[
 U(w)=T_m(w)+\beta_mW(w).
\]
Since $\beta_j=400K(8K+1)q^{j+1}$ and
$q^j\le e^{-j/(8K+1)}$, we have
\[
 m=O\!\left(K\log\frac{2K}{\eta}\right).
\]
By \cref{lem:taylor-construction}, $U$ has rational coefficients,
degree at most $m$, and the claimed construction and evaluation costs.

To prove convexity, write $L(y)=\log Z_G(e^y)$, and let $X_e(M)$
indicate that $e\in M$. Differentiating the finite sum defining
$Z_G$ gives
\[
 \partial_e L(y)=\E_{\mu_{e^y}}X_e,
 \qquad
 \partial_f\partial_e L(y)=\Cov_{\mu_{e^y}}(X_e,X_f).
\]
Consequently, for every $a\in\R^E$,
\[
 a^{\mathsf T}\nabla^2L(y)a
 =\Var_{\mu_{e^y}}\!\left(\sum_e a_eX_e\right)\ge0.
\]
Also, $\nabla_y^2W(e^y)=\diag(e^{y_e}:e\in E)$. Since
$U(e^y)=L(y)-\mathcal E_m(e^y)+\beta_mW(e^y)$, the Hessian bound
in \cref{lem:taylor-error-bounds} gives
\begin{align*}
 a^{\mathsf T}\nabla_y^2[U(e^y)]a
 &=a^{\mathsf T}\nabla^2L(y)a
   -a^{\mathsf T}\nabla_y^2\mathcal E_m(e^y)a
   +\beta_m\sum_e e^{y_e}a_e^2
 \ge0.
\end{align*}
Thus $y\mapsto U(e^y)$ is convex on the convex set $\cD_K$.

For the error bound, use
$U(e^y)-L(y)=\beta_mW(e^y)-\mathcal E_m(e^y)$
and \cref{eq:taylor-value-error} to obtain
\[
 0\le U(e^y)-L(y)\le2\beta_mW(e^y).
\]
Since
\[
 2W(e^y)=\sum_{v\in V}\sum_{e\ni v}e^{y_e}\le NK,
\]
it follows that
\[
 0\le U(e^y)-L(y)\le\beta_mNK\le\eta N/2.
\]
This proves the theorem.
\end{proof}

\subsection{Path trees}

We now turn to the proof of \cref{lem:taylor-error-bounds}. We will bound
derivatives of $\log Z_G(\tau(z)e^y)$ by expressing them through
ratios of partition functions.

For an induced subgraph $H$ of $G$, a vertex $v\in V(H)$, and
$s\in\C$, define 
\[
 R_{H,v}(sw)=\frac{Z_{H-v}(sw)}{Z_H(sw)}
\]
whenever the denominator is nonzero. For $e=uv$, removing $e$ from
a matching containing it leaves a matching of $G-u-v$. Thus
\[
 \partial_e Z_G(sw)
 =\sum_{\substack{M\in\cM(G)\\e\in M}}(sw)^M
 =sw_eZ_{G-u-v}(sw),
\]
and hence
\begin{equation*}
 \partial_e\log Z_G(sw)
 =sw_e\frac{Z_{G-u-v}(sw)}{Z_G(sw)}
 =sw_eR_{G,u}(sw)R_{G-u,v}(sw).
\end{equation*}
This identity holds wherever the expressions are defined.
To bound these ratios, we represent them using a matrix associated
with a tree.

For an induced subgraph $H$ and $v\in V(H)$, the {path tree}
$\mathcal T(H,v)$ has as its vertices the simple paths in $H$
starting at $v$. Two paths are adjacent if one is obtained from
the other by adding one edge at the end. The root is $r=(v)$,
the path consisting only of $v$. Every other path has a unique
parent, obtained by deleting its last edge, so this graph is
indeed a tree. See \cref{fig:path-tree}.

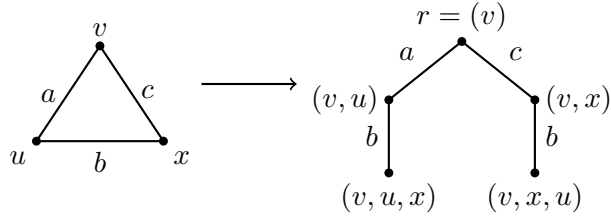
\begin{figure}[htbp]
\centering
\begin{tikzpicture}[x=1cm,y=1cm,scale=.84]
  \coordinate (v) at (-3.25,.65);
  \coordinate (u) at (-4.25,-.85);
  \coordinate (x) at (-2.25,-.85);
  \draw[thick] (v)--node[left,pos=.52] {$a$} (u);
  \draw[thick] (u)--node[below] {$b$} (x);
  \draw[thick] (x)--node[right,pos=.52] {$c$} (v);
  \fill (v) circle (2pt) node[above] {$v$};
  \fill (u) circle (2pt) node[below left] {$u$};
  \fill (x) circle (2pt) node[below right] {$x$};

  \draw[->,thick] (-1.65,.05)--(-.15,.05);

  \coordinate (r) at (2.45,.72);
  \coordinate (vu) at (1.30,-.20);
  \coordinate (vx) at (3.60,-.20);
  \coordinate (vux) at (1.30,-1.35);
  \coordinate (vxu) at (3.60,-1.35);
  \draw[thick] (r)--node[above left] {$a$} (vu);
  \draw[thick] (r)--node[above right] {$c$} (vx);
  \draw[thick] (vu)--node[left] {$b$} (vux);
  \draw[thick] (vx)--node[right] {$b$} (vxu);
  \fill (r) circle (2pt) node[above] {$r=(v)$};
  \fill (vu) circle (2pt) node[left] {$(v,u)$};
  \fill (vx) circle (2pt) node[right] {$(v,x)$};
  \fill (vux) circle (2pt) node[below] {$(v,u,x)$};
  \fill (vxu) circle (2pt) node[below] {$(v,x,u)$};

\end{tikzpicture}
\caption{The graph $H$ (left) and its path tree $\mathcal T(H,v)$
rooted at $r=(v)$ (right). Edge labels indicate the corresponding
edges of $H$.}
\label{fig:path-tree}
\end{figure}

Let $B$ be the real symmetric matrix indexed by the vertices of
$\mathcal T(H,v)$, with
\[
 B_{pp'}=
 \begin{cases}
  \sqrt{w_e},&\text{if one of $p,p'$ extends the other by the edge $e$},\\
  0,&\text{otherwise}.
 \end{cases}
\]

The following lemma expresses $R_{H,v}$ in terms of $B$ and shows
that $\norm B_2\le2\sqrt K$ whenever $y\in\cD_K$.

\begin{lemma}\label{lem:path-tree-resolvent}
For every induced subgraph $H$ of $G$ and $v\in V(H)$, let $B$ and
$r$ be as above. Then
\begin{equation}\label{eq:path-tree-resolvent}
 R_{H,v}(sw)=\bigl[(I+sB^2)^{-1}\bigr]_{rr}
\end{equation}
whenever both sides are defined. If $w=e^y$ with $y\in\cD_K$, then
\begin{equation*}
 \norm{B}_2\le2\sqrt K.
\end{equation*}
\end{lemma}

\begin{proof}
We prove the weighted resolvent form of Godsil's
path-tree identity~\cite{Godsil} directly for completeness.
Every matching of $H$ either leaves $v$ unmatched or contains
exactly one edge $vu$. Thus
\[
 Z_H(sw)=Z_{H-v}(sw)
       +s\sum_{u\sim v}w_{vu}Z_{H-v-u}(sw),
\]
and hence
\[
 R_{H,v}(sw)
 =\left(1+s\sum_{u\sim v}w_{vu}R_{H-v,u}(sw)\right)^{-1}.
\]
We now show that the diagonal entry of $(I-\zeta B_{H,v})^{-1}$
corresponding to the root satisfies the same recurrence,
with $s=-\zeta^2$.

Write $B_{H,v}$ for the weighted adjacency matrix of
$\mathcal T(H,v)$, and set
\[
 Q_{H,v}(\zeta)
 =\bigl[(I-\zeta B_{H,v})^{-1}\bigr]_{rr}.
\]
For now, take $\zeta$ sufficiently close to zero that all
inverses below exist.
Deleting the root of $\mathcal T(H,v)$ leaves one component
for each neighbor $u$ of $v$.
Removing the initial vertex $v$ from each path identifies
this component with $\mathcal T(H-v,u)$. Order the matrix indices with the root first, followed by
these components. Then
\[
 I-\zeta B_{H,v}
 =\begin{pmatrix}
    1&-\zeta b^{\mathsf T}\\
    -\zeta b&M
   \end{pmatrix},
\]
where $M$ is block diagonal with blocks $I-\zeta B_{H-v,u}$
for $u\sim v$, and $b_{(v,u)}=\sqrt{w_{vu}}$, with all other
entries of $b$ zero.
Taking the Schur complement of $M$ gives
\[
 Q_{H,v}(\zeta)
 =\left(1-\zeta^2b^{\mathsf T}M^{-1}b\right)^{-1}.
\]
By the block structure of $M$ and the definition of $b$,
\[
 b^{\mathsf T}M^{-1}b
 =\sum_{u\sim v}w_{vu}Q_{H-v,u}(\zeta),
\]
so
\[
 Q_{H,v}(\zeta)
 =\left(1-\zeta^2\sum_{u\sim v}
                   w_{vu}Q_{H-v,u}(\zeta)\right)^{-1}.
\]
Observe that this is also the recurrence satisfied by
$R_{H,v}(-\zeta^2w)$.
Since $Q_{H,v}(\zeta)=R_{H,v}(-\zeta^2w)=1$ when $v$ is isolated,
induction on $|V(H)|$ gives
\[
 Q_{H,v}(\zeta)=R_{H,v}(-\zeta^2w).
\]

To express this in terms of $B^2$, let $D$ be diagonal,
with entries $1$ and $-1$ on the two bipartition classes
of the path tree.
Since $DBD=-B$ and $D^2=I$,
\[
 (I+\zeta B)^{-1}=D(I-\zeta B)^{-1}D.
\]
The two inverses therefore have the same $rr$ entry.
Averaging gives
\begin{align*}
 R_{H,v}(-\zeta^2w)
 &=\frac12\bigl[(I-\zeta B)^{-1}
                  +(I+\zeta B)^{-1}\bigr]_{rr}\\
 &=\bigl[(I-\zeta^2B^2)^{-1}\bigr]_{rr}.
\end{align*}
Setting $s=-\zeta^2$ proves \cref{eq:path-tree-resolvent}
near $s=0$. Both sides are rational functions of $s$,
so the identity holds wherever both sides are defined.

For the norm bound, write $B=P+P^{\mathsf T}$, where
\[
 P_{pp'}=
 \begin{cases}
  B_{pp'},&\text{if $p$ is the parent of $p'$},\\
  0,&\text{otherwise}.
 \end{cases}
\]
Each vertex has at most one parent, so distinct rows of $P$
have disjoint supports. Thus $PP^{\mathsf T}$ is diagonal.
If a path $p$ ends at $u$, its children extend it along
distinct edges incident to $u$, giving
\[
 (PP^{\mathsf T})_{pp}
 =\sum_{p'\text{ child of }p}B_{pp'}^2
 \le\sum_{\substack{e\in E(H)\\e\ni u}}w_e
 \le K,
\]
where the last inequality uses $w=e^y$ and $y\in\cD_K$.
Hence $PP^{\mathsf T}\preceq KI$ and $\norm P_2\le\sqrt K$.
It follows that
\[
 \norm B_2
 \le\norm P_2+\norm{P^{\mathsf T}}_2
 \le2\sqrt K.
 \qedhere
\]
\end{proof}

The path tree is used only in the proof; the algorithm computes
$T_m$ by the matching enumeration in
\cref{lem:taylor-construction}. Since $B$ is real symmetric,
the norm bound in \cref{lem:path-tree-resolvent} gives
\[
 \operatorname{spec}(B^2)\subseteq[0,4K]
\]
for every induced subgraph $H$ and root $v\in V(H)$, where
$\operatorname{spec}(B^2)$ denotes the spectrum of $B^2$.
We will use this inclusion to bound $(I+\tau(z)B^2)^{-1}$ for
complex $z$.

\subsection{The approximation estimates}
Fix $y\in\cD_K$ and write $w=e^y$. We will estimate the Taylor
coefficients on a circle of radius
\[
 R=\frac{1+z_\star}{2},
\]
so that $z_\star<R<1$. From the definitions of $\tau$ and $z_\star$,
we have
\begin{equation}\label{eq:interpolation-contour-bounds}
 \begin{aligned}
  \frac{z_\star}{R}&=\frac{8K}{8K+1}=q,\\
  |\tau(z)|&\le\frac{\rho R}{1-R}
             =2+\frac1{4K}\le\frac94
             \qquad(|z|\le R).
 \end{aligned}
\end{equation}
The change of variables is illustrated in \cref{fig:interpolation}.

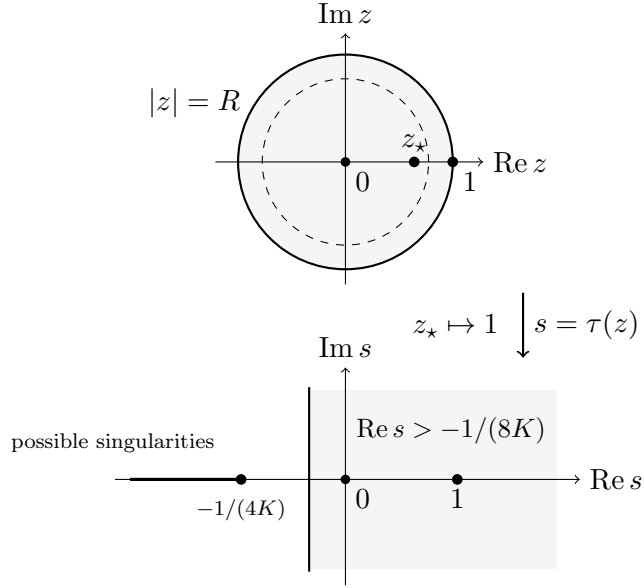
\begin{figure}[htbp]
\centering
\begin{tikzpicture}[x=1cm,y=1cm]
  \begin{scope}[shift={(0,2.45)}]
    \fill[gray!8] (0,0) circle (1.42);
    \draw[thick] (0,0) circle (1.42);
    \draw[dashed] (0,0) circle (1.10);
    \draw[->] (-1.72,0)--(1.82,0)
       node[right] {$\operatorname{Re}z$};
    \draw[->] (0,-1.62)--(0,1.70)
       node[above] {$\operatorname{Im}z$};
    \fill (0,0) circle (1.8pt) node[below right] {$0$};
    \fill (.91,0) circle (2pt) node[above] {$z_\star$};
    \fill (1.42,0) circle (2pt) node[below right] {$1$};
    \node[fill=white,inner sep=1.5pt] at (-2,.78) {$|z|=R$};
  \end{scope}

  \draw[->,thick] (2.35,.72)--(2.35,-.14)
       node[midway,right] {$s=\tau(z)$};
  \node[anchor=east] at (2.16,.28) {$z_\star\mapsto1$};

  \begin{scope}[shift={(0,-1.75)}]
    \fill[gray!8] (-.48,-1.18) rectangle (2.80,1.18);
    \draw[thick] (-.48,-1.22)--(-.48,1.22);
    \draw[->] (-3.05,0)--(3.10,0)
       node[right] {$\operatorname{Re}s$};
    \draw[->] (0,-1.40)--(0,1.48)
       node[above] {$\operatorname{Im}s$};
    \draw[very thick] (-2.85,0)--(-1.38,0);
    \fill (-1.38,0) circle (2pt);
    \fill (0,0) circle (1.8pt) node[below right] {$0$};
    \fill (1.48,0) circle (2pt) node[below] {$1$};
    \node[font=\small] at (1.40,.65)
       {$\operatorname{Re}s>-1/(8K)$};
\node[anchor=east,font=\scriptsize] at (-1.60,.48)
   {possible singularities};
    \node[anchor=north,font=\scriptsize] at (-1.38,-.12)
       {$-1/(4K)$};
  \end{scope}
\end{tikzpicture}
\caption{The map $s=\tau(z)$ sends the unit disk to the half-plane
$\operatorname{Re}s>-1/(8K)$, whereas $I+sB^2$ can be singular only
for real $s\le-1/(4K)$. The Taylor series is evaluated at $z_\star$
and estimated on the larger circle $|z|=R$. The figure is schematic.}
\label{fig:interpolation}
\end{figure}

To control the Hessian of the approximation error, we will also
need bounds on certain second derivatives. For
$\sigma\in\{-1,1\}^E$, define
\[
 D_\sigma=\sum_{f\in E}\sigma_f\partial_f.
\]
We will use bounds on $D_\sigma\partial_eF_y(z)$ to bound the sum
of the absolute values of the entries in each row of
$\nabla_y^2\mathcal E_m(e^y)$.

\begin{lemma}\label{lem:analytic-derivatives}
For every induced subgraph $H$ of $G$,
$Z_H(\tau(z)w)\ne0$ for $|z|<1$. In particular, there is a
holomorphic logarithm
\[
 F_y(z)=\log Z_G(\tau(z)e^y)
 \qquad(|z|<1)
\]
with $F_y(0)=0$. For every edge $e$, every
$\sigma\in\{-1,1\}^E$, and $|z|\le R$,
\begin{align}
 |\partial_eF_y(z)|&\le9e^{y_e},
 \label{eq:analytic-gradient-bound}\\
 |D_\sigma\partial_eF_y(z)|&\le400K e^{y_e}.
 \label{eq:analytic-signed-hessian-bound}
\end{align}
\end{lemma}

\begin{proof}
For $|z|<1$,
\[
 \operatorname{Re}\frac{z}{1-z}
 =\frac12\left(\frac{1-|z|^2}{|1-z|^2}-1\right)>-\frac12,
\]
so $\operatorname{Re}\tau(z)>-1/(8K)$. Let $B$ be the matrix
associated with any path tree $\mathcal T(H,v)$.
By \cref{lem:path-tree-resolvent}, $\norm B_2\le2\sqrt K$,
so every eigenvalue $a$ of $B^2$ lies in $[0,4K]$. Hence
\[
 \operatorname{Re}(1+\tau(z)a)
 \ge1-\frac{a}{8K}\ge\frac12.
\]
Since $B^2$ is real symmetric,
\begin{equation*}
 \norm{(I+\tau(z)B^2)^{-1}}_2
 =\max_{a\in\operatorname{spec}(B^2)}
       \frac1{|1+\tau(z)a|}
 \le2
 \qquad(|z|<1).
\end{equation*}

We next prove that $Z_H(\tau(z)w)\ne0$ for every induced subgraph
$H$, by induction on $|V(H)|$. The empty graph has partition
function $1$. For nonempty $H$, choose $v\in V(H)$ and let $B$
be the matrix of $\mathcal T(H,v)$, with root $r$.
By \cref{eq:path-tree-resolvent},
\[
 Z_H(sw)\bigl[(I+sB^2)^{-1}\bigr]_{rr}=Z_{H-v}(sw)
\]
as an identity of rational functions in $s$. By the inverse bound above,
the inverse exists at $s=\tau(z)$, so the identity holds there.
Its right-hand side is nonzero by induction, and therefore
$Z_H(\tau(z)w)\ne0$.

The unit disk is simply connected, so $Z_G(\tau(z)w)$ admits a
holomorphic logarithm there. Since $Z_G(\tau(0)w)=1$, we may choose
it to satisfy $F_y(0)=0$. By continuity and compactness,
$Z_G(\tau(z)e^{y'})$ remains nonzero on a common disk containing
$|z|\le R$ for all $y'$ sufficiently close to $y$. On this disk,
choose each logarithm to satisfy $F_{y'}(0)=0$. These logarithms
depend smoothly on $y'$, so we may differentiate with respect
to $y$, using the usual rule
\[
 \partial_eF_y(z)
 =\frac{\partial_e Z_G(\tau(z)e^y)}{Z_G(\tau(z)e^y)}.
\]

For any induced subgraph $H$ and vertex $v\in V(H)$, the path-tree
representation \cref{eq:path-tree-resolvent} and the inverse norm bound
above give
\begin{equation*}
 |R_{H,v}(\tau(z)w)|
 =\left|\bigl[(I+\tau(z)B^2)^{-1}\bigr]_{rr}\right|
 \le2
 \qquad(|z|<1).
\end{equation*}
For $e=uv$, differentiating the partition function gives
\[
 \partial_eF_y(z)
 =\tau(z)w_e\frac{Z_{G-u-v}(\tau(z)w)}{Z_G(\tau(z)w)}
 =\tau(z)w_eR_{G,u}(\tau(z)w)R_{G-u,v}(\tau(z)w).
\]
Since $|\tau(z)|\le9/4$ on $|z|\le R$ by
\cref{eq:interpolation-contour-bounds}, it follows that
\[
 |\partial_eF_y(z)|
 \le\frac94w_e\cdot2\cdot2=9w_e
 \qquad(|z|\le R).
\]
This proves \cref{eq:analytic-gradient-bound}.

To bound $D_\sigma\partial_eF_y(z)$, we first differentiate the
matrix expression for $R_{H,v}$. Write $B_\sigma=D_\sigma B$,
where the derivative is taken entrywise. If the edge joining
$p$ and $p'$ corresponds to $f\in E(H)$, then
\[
 (B_\sigma)_{pp'}
 =D_\sigma(e^{y_f/2})
 =\frac{\sigma_f}{2}B_{pp'}.
\]
As in the path-tree proof, write $B=P+P^{\mathsf T}$, where
$P_{pp'}=B_{pp'}$ when $p$ is the parent of $p'$, and all other
entries of $P$ are zero. Let $P_\sigma=D_\sigma P$, so that
\[
 B_\sigma=P_\sigma+P_\sigma^{\mathsf T}.
\]
As before, $P_\sigma P_\sigma^{\mathsf T}$ is diagonal.
Its diagonal entries are one quarter of those of
$PP^{\mathsf T}$, because each entry of $P_\sigma$ has half
the absolute value of the corresponding entry of $P$. Hence
\[
 P_\sigma P_\sigma^{\mathsf T}
 =\frac14PP^{\mathsf T}\preceq\frac K4I.
\]
It follows that
\begin{equation*}
 \norm{B_\sigma}_2
 \le\norm{P_\sigma}_2+\norm{P_\sigma^{\mathsf T}}_2
 \le\sqrt K.
\end{equation*}

Differentiating the identity
$(I+\tau(z)B^2)(I+\tau(z)B^2)^{-1}=I$ gives
\[
 D_\sigma(I+\tau(z)B^2)^{-1}
 =-(I+\tau(z)B^2)^{-1}
   \tau(z)(B_\sigma B+BB_\sigma)
   (I+\tau(z)B^2)^{-1}.
\]
The ratio $R_{H,v}(\tau(z)w)$ is the $rr$ entry of this inverse.
Therefore, using the bounds on the inverse, $\tau(z)$, $B$,
and $B_\sigma$, we obtain
\begin{align*}
 |D_\sigma R_{H,v}(\tau(z)w)|
 &\le
 2|\tau(z)|\norm{(I+\tau(z)B^2)^{-1}}_2^2
   \norm B_2\norm{B_\sigma}_2\\
 &\le2\cdot\frac94\cdot2^2\cdot2\sqrt K\cdot\sqrt K
 =36K
 \qquad(|z|\le R).
\end{align*}

Finally, for $e=uv$, let
\[
 R_1=R_{G,u}(\tau(z)w),
 \qquad
 R_2=R_{G-u,v}(\tau(z)w).
\]
Since $\partial_eF_y(z)=\tau(z)w_eR_1R_2$ and
$D_\sigma w_e=\sigma_ew_e$, the product rule gives
\[
 D_\sigma\partial_eF_y(z)
 =\tau(z)w_e\bigl(
   \sigma_eR_1R_2+(D_\sigma R_1)R_2+R_1(D_\sigma R_2)
   \bigr).
\]
On $|z|\le R$, each ratio has absolute value at most $2$
by the bounds established above, and each of its signed derivatives
has absolute value at most $36K$ by
the derivative bound above. Thus,
\[
 |D_\sigma\partial_eF_y(z)|
 \le\frac94w_e\bigl(4+36K\cdot2+2\cdot36K\bigr)
 =(9+324K)w_e
 \le400K w_e,
\]
where the last inequality uses $K\ge1$. This proves
\cref{eq:analytic-signed-hessian-bound}.
\end{proof}

With these bounds in hand, we return to the approximation error.

\begin{proof}[Proof of \cref{lem:taylor-error-bounds}]
Write the Taylor expansion of $F_y$ as
\[
 F_y(z)=\sum_{\ell\ge0}f_\ell(y)z^\ell,
 \qquad f_0(y)=0,\quad f_\ell(y)=\lambda_\ell(e^y)\ (\ell\ge1).
\]
By Cauchy's integral formula,
\[
 f_\ell(y)
 =\frac{1}{2\pi i}\int_{|z|=R}
       \frac{F_y(z)}{z^{\ell+1}}\,dz.
\]
In the preceding proof, we showed that the logarithms $F_{y'}$
can be chosen on a common disk containing this contour for all
$y'$ sufficiently close to $y$, and that they depend smoothly
on $y'$. We may therefore differentiate under the integral:
\[
 \partial_e f_\ell(y)
 =\frac{1}{2\pi i}\int_{|z|=R}
       \frac{\partial_eF_y(z)}{z^{\ell+1}}\,dz.
\]
Thus $\partial_e f_\ell(y)$ is the coefficient of $z^\ell$
in the Taylor expansion of $\partial_eF_y(z)$. The same argument
applies to second derivatives.

The approximation uses the terms of $F_y$ through degree $m$,
evaluated at $z_\star$. Therefore,
\[
 |\partial_e\mathcal E_m(e^y)|
 =\left|\partial_eF_y(z_\star)
       -\sum_{\ell=0}^m\partial_e f_\ell(y)z_\star^\ell\right|
 \le\sum_{\ell>m}|\partial_e f_\ell(y)|z_\star^\ell.
\]
By \cref{eq:analytic-gradient-bound}, we have
$|\partial_eF_y(z)|\le9w_e$ on $|z|=R$. Applying this bound
in the contour integral gives
\[
 |\partial_e f_\ell(y)|\le9w_eR^{-\ell}.
\]
Recall from \cref{eq:interpolation-contour-bounds} that $z_\star/R=q<1$.
Hence,
\[
 |\partial_e\mathcal E_m(e^y)|
 \le9w_e\sum_{\ell>m}q^\ell
 =\frac{9q^{m+1}}{1-q}w_e
 \le\beta_mw_e.
\]
This proves \cref{eq:taylor-gradient-error}.

For the Hessian estimate, we first apply the same argument to
$D_\sigma\partial_eF_y(z)$. By
\cref{eq:analytic-signed-hessian-bound}, its absolute value on
$|z|=R$ is at most $400K w_e$. Hence
\[
 |D_\sigma\partial_e f_\ell(y)|
 \le400K w_eR^{-\ell}.
\]
It follows that, for every $\sigma\in\{-1,1\}^E$,
\[
 |D_\sigma\partial_e\mathcal E_m(e^y)|
 \le\sum_{\ell>m}|D_\sigma\partial_e f_\ell(y)|z_\star^\ell
 \le\frac{400Kq^{m+1}}{1-q}w_e
 =\beta_mw_e.
\]
We now choose the signs to turn this into a bound on a row of
the Hessian. Write
\[
 H=(H_{ef})=\nabla_y^2\mathcal E_m(e^y),
\]
which is a real symmetric matrix. Fix a row $e$ and choose
$\sigma_f\in\{-1,1\}$ so that $\sigma_fH_{ef}=|H_{ef}|$;
either sign may be chosen when $H_{ef}=0$. By the definition
of $D_\sigma$,
\[
 D_\sigma\partial_e\mathcal E_m(e^y)
 =\sum_f\sigma_fH_{ef}
 =\sum_f|H_{ef}|.
\]
The bound just proved therefore gives
\begin{equation*}
 \sum_f|H_{ef}|\le\beta_mw_e.
\end{equation*}
To deduce the required quadratic-form bound, let $a\in\R^E$.
Using $2|a_ea_f|\le a_e^2+a_f^2$, we obtain
\begin{align*}
 |a^{\mathsf T}Ha|
 &\le\sum_{e,f}|H_{ef}|\,|a_ea_f| \le\frac12\sum_{e,f}|H_{ef}|(a_e^2+a_f^2) =\sum_ea_e^2\sum_f|H_{ef}|.
\end{align*}
For the last equality, symmetry of $H$ shows that the terms
containing $a_f^2$ have the same sum as those containing $a_e^2$.
Applying this row-sum bound now gives
\[
 |a^{\mathsf T}Ha|
 \le\beta_m\sum_e w_ea_e^2.
\]
This proves the Hessian estimate.

Finally, we bound $|\mathcal E_m(w)|$ by integrating the derivative
along the segment from $0$ to $w$. For $0<s\le1$, let
\[
 y(s)=y+(\log s)\one,
 \qquad\text{so that}\qquad e^{y(s)}=sw.
\]
Scaling all weights by $s\le1$ preserves the defining inequalities
of $\cD_K$, so $y(s)\in\cD_K$. Since $y'_e(s)=1/s$ for every
edge $e$, the chain rule gives
\[
 \frac{d}{ds}\mathcal E_m(sw)
 =\frac1s\sum_e\partial_e\mathcal E_m(e^{y(s)}).
\]
The gradient bound \cref{eq:taylor-gradient-error}, applied at $y(s)$,
says that
\[
 |\partial_e\mathcal E_m(e^{y(s)})|
 \le \beta_m e^{y_e(s)}=\beta_msw_e.
\]
Hence,
\[
 \left|\frac{d}{ds}\mathcal E_m(sw)\right|
 \le\frac1s\sum_e \beta_msw_e
 =\beta_mW(w).
\]
Integrating from $\delta$ to $1$, for $0<\delta<1$, gives
\[
 |\mathcal E_m(w)-\mathcal E_m(\delta w)|
 \le(1-\delta)\beta_mW(w).
\]
Since $Z_G(0)=1$ and $T_m(0)=0$, continuity gives
$\mathcal E_m(\delta w)\to0$ as $\delta\downarrow0$.
Letting $\delta\downarrow0$, we conclude that
\[
 |\mathcal E_m(w)|\le \beta_mW(w),
\]
which proves \cref{eq:taylor-value-error}.
\end{proof}
\section{The algorithm}\label{sec:algorithm}

Our algorithm is described in \cref{alg:permanent}.
The subroutine \textsc{ScaleAndRound} scales and rounds $A$ to a matrix
$C$ whose entries have bit length polynomial in $n$ and $\log(1/\varepsilon)$.
The subroutine \textsc{MinimizeSurrogate} uses the polynomial from
\cref{sec:convex-approximation} to approximate $F_C(t)$ by convex minimization.

For $0<\varepsilon\le1$, choose $\alpha=2^{-s}$ for an integer $s$ with
$\varepsilon/8<\alpha\le\varepsilon/4$. The parameters used below are
\begin{equation}\label{eq:algorithm-parameters}
 \eta=\frac{\alpha}{16},\qquad \delta=\frac{\eta}{s+1},\qquad
 t=1-\delta,\qquad K=4\delta^{-2}.
\end{equation}

\begin{algorithm}[htbp]
\caption{Deterministic approximation of the nonnegative permanent}
\label{alg:permanent}
\begin{algorithmic}[1]
\Require $A\in\R_{\ge0}^{n\times n}$ and $0<\varepsilon\le1$.
\Ensure $Q$ satisfying $\per A\le Q\le(1+\varepsilon)^n\per A$.
\State If $n=0$, \Return $1$; if $n=1$, \Return $A_{11}$.

\State If the support graph has no perfect matching, \Return $0$.

\State Set $A_{uv}=0$ for every support edge $uv$ belonging to no perfect matching.

\State Choose $\alpha=2^{-s}$ with
$\varepsilon/8<\alpha\le\varepsilon/4$, and set
$t,\eta,K$ by~\cref{eq:algorithm-parameters}.

\State $(S,C,p)\gets\Call{ScaleAndRound}{A,\alpha}$
\Comment{\cref{lem:scale-and-round}}

\State Construct $U$ for the support of $C$ using
\cref{thm:surrogate} with parameters $K,\eta$.

\State $\widehat F\gets\Call{MinimizeSurrogate}{C,U,t,\eta,p,s}$
\Comment{error $2\eta n$, rounded up}

\State Compute a rational $Q_C$ with
$e^{\wh{F}/t}\le Q_C\le e^{\wh{F}/t}+\eta n^{-n}$.
\label{line:exponential}
\State \Return $Q_C/S$.

\end{algorithmic}
\end{algorithm}

\subsection{\textsc{ScaleAndRound}}\label{sec:scale-and-round}
The following lemma combines the preprocessing result of Linial,
Samorodnitsky, and Wigderson~\cite{LSW} with a simple rounding step.

\begin{lemma}\label{lem:scale-and-round}
Let $n\ge2$, let $A\in\R_{\ge0}^{n\times n}$ have positive permanent,
and let $\alpha=2^{-s}$ for an integer $s\ge1$.
There is an algorithm \emph{\textsc{ScaleAndRound}} that, given
$A$ and $\alpha$, returns $S>0$, an integer $p\ge0$, and a matrix
$C\in[0,1]^{n\times n}$ with the same support as $A$, satisfying
\begin{equation*}
 S\per A\le\per C\le(1+\alpha)S\per A,
 \qquad \per C\ge1.
\end{equation*}
The entries of $C$ are multiples of $2^{-p}$ for
$p=O(n\log n+s)$, and the running time is
$\operatorname{poly}(n,s)$.
\end{lemma}

\begin{proof}
The LSW preprocessing procedure
\cite[Section~3.1 and Theorem~3.1]{LSW}
computes a positive diagonal matrix $Y$ and a permutation $\sigma$
such that $(AY)_{i\sigma(i)}\ge(AY)_{ij}$ for every $i,j$.
Since $\per A>0$, each selected entry is positive. Set
\[
 D_2=Y,\qquad (D_1)_{ii}=(AY)_{i\sigma(i)}^{-1},\qquad B=D_1AD_2.
\]
Then $B\in[0,1]^{n\times n}$ and $B_{i\sigma(i)}=1$ for every $i$,
so $\per B\ge1$.
Every term in the permanent contains one entry from each row
and each column, so
\[
 \per B=S\per A,
 \qquad
 S:=\left(\prod_{i=1}^n(D_1)_{ii}\right)
    \left(\prod_{j=1}^n(D_2)_{jj}\right).
\]

For an integer $p$ to be chosen later, round each entry of $B$ upward
to the nearest multiple of $2^{-p}$, obtaining $C$.
This preserves the support and keeps all entries in $[0,1]$.
Each summand in the definition of the permanent increases by at most $n2^{-p}$; to see this, replace each of the $n$ factors one at a time, using that all other factors
are at most one. Summing over the $n!$ permutations gives
\[
 0\le\per C-\per B\le (n 2^{-p}) n!.
\]
Choose the smallest nonnegative integer $p$ such that
$(n 2^{-p})n!\le\alpha$.
Then $p=O(n\log n+s)$, and since $\per B\ge1$,
\[
 \per B\le\per C\le\per B+\alpha
 \le(1+\alpha)\per B.
\]
Together with $\per B=S\per A$, this proves the lemma.

The preprocessing of LSW takes polynomial time.
It takes $O(n^2p)$ additional operations to (1) find $p$ by successive halving, and (2) round each entry of $B$ by binary
search.
For rational input, the intermediate bit lengths are polynomial
in the input length and $s$, by the LSW bound
\cite[proof of Theorem~3.4]{LSW} and the bound on $p$.
\end{proof}

\subsection{\textsc{MinimizeSurrogate}}
\label{sec:surrogate-minimization}

This subroutine approximates
$F_C(t)=\inf_\theta\Phi_{C,t}(\theta)$ for the matrix $C$
produced by \cref{alg:permanent}. Let $G=(V,E)$ be the support graph
of $C$, with $|V|=2n$, and let $p$ be the precision returned by
\textsc{ScaleAndRound}. We first bound the norm of a minimizer.

\begin{lemma}\label{lem:field-bound}
The infimum of $\Phi_{C,t}$ is attained at a vector $\theta^*$
satisfying
\begin{equation}\label{eq:field-radius}
 \norm{\theta^*}_\infty\le R_0:=\frac{2n^2(n+p)}{\delta}.
\end{equation}
\end{lemma}

\begin{proof}
Every edge in the support of $C$ belongs to a perfect matching,
so \cref{lem:balanced-gibbs} guarantees that the infimum is attained
at some $\theta^*$.
Comparing with $\theta=0$ gives
\[
 \Phi_{C,t}(\theta^*)\le\Phi_{C,t}(0)
 =\log Z_G(C)\le n\log(n+1),
\]
where the last inequality uses $C_e\le1$ and the fact that
$G$ has at most $(n+1)^n$ matchings.
Each matching has at most $n$ edges, so
$2^{-np}\le C^M\le1$ for every matching $M$, and hence,
$|\log C^M|\le np\log2$.
Since $t=1-\delta\ge1/2$, \cref{eq:edge-coercivity} gives
\[
 \Phi_{C,t}(\theta)
 \ge \delta\max_{uv\in E}|\theta_u+\theta_v|-np\log2.
\]
Applying this at $\theta^*$ and using
$\Phi_{C,t}(\theta^*)\le n\log(n+1)$, we obtain
\begin{equation}\label{eq:edge-sum-bound}
 \max_{uv\in E}|\theta_u^*+\theta_v^*|
 \le\frac{n\log(n+1)+np\log2}{\delta}=:H.
\end{equation}

We now bound the coordinates $\theta^*_u$.
Since $G$ has a perfect matching, every connected component
contains equally many left and right vertices.
In each component $L'\sqcup R'$, choose $u_0\in L'$ arbitrarily.
The transformation
\[
 \theta_u^*\mapsto\theta_u^*-\theta^*_{u_0} \quad(u\in L'),
 \qquad
 \theta_v^*\mapsto\theta_v^*+\theta^*_{u_0} \quad(v\in R')
\]
leaves $\Phi_{C,t}$ unchanged and makes $\theta^*_{u_0}=0$.
Along any simple path $u_0,u_1,\ldots,u_k$,
\cref{eq:edge-sum-bound} gives
\[
 |\theta^*_{u_i}|
 \le |\theta^*_{u_{i-1}}|
    +|\theta^*_{u_{i-1}}+\theta^*_{u_i}|
 \le |\theta^*_{u_{i-1}}|+H.
\]
Thus $|\theta^*_{u_k}|\le kH<2nH$.
Every vertex can be reached by such a path.
Using $\log(n+1)\le n$ and $\log2\le1$, we conclude that
\[
 \norm{\theta^*}_\infty<2nH
 \le\frac{2n^2(n+p)}{\delta}=R_0. \qedhere
\]
\end{proof}

Write $d_v(\theta)=\sum_{e\ni v}w_C(\theta)_e$ for the total
weight of the edges incident to $v$.
Recall from~\cref{eq:degree-bound} that every minimizer
$\theta^*$ of $\Phi_{C,t}$ satisfies
\[
 d_v(\theta^*)\le\frac{t}{(1-t)^2}
 =\frac{t}{\delta^2}<\frac K2
 \qquad\text{for every }v.
\]
The minimizer from \cref{lem:field-bound} also satisfies
$\norm{\theta^*}_\infty\le R_0$.
We may therefore restrict the minimization to
\begin{equation*}
 \mathcal{D}=\{\theta\in\R^{2n}:\norm{\theta}_\infty\le R_0,\ 
                  d_v(\theta)\le K/2\text{ for every }v\}.
\end{equation*}

We now replace $\log Z_G$ by the polynomial $U$ constructed in
\cref{sec:convex-approximation}, with $m$ chosen by~\cref{eq:surrogate-degree},
and minimize
\begin{equation*}
 f(\theta)=U(w_C(\theta))-t\sum_v\theta_v.
\end{equation*}
For $uv\in E$, set
$y_C(\theta)_{uv}=\log C_{uv}+\theta_u+\theta_v$.
This map is affine, $w_C(\theta)=e^{y_C(\theta)}$, and
$y_C(\mathcal D)\subseteq\cD_{K/2}\subseteq\cD_K$.
Moreover, $\mathcal D$ is convex, being the intersection of a box
with the affine preimage of $\cD_{K/2}$.
Thus \cref{thm:surrogate}, applied with $N=2n$, shows that $f$ is
convex on $\mathcal D$ and satisfies
$\Phi_{C,t}\le f\le\Phi_{C,t}+\eta n$ there.
Taking minima over $\mathcal{D}$ gives
\begin{equation}\label{eq:surrogate-minimum}
 F_C(t)\le\min_{\theta\in\mathcal{D}}f(\theta)\le F_C(t)+\eta n.
\end{equation}

The following lemma shows that this minimum can be efficiently
approximated with an additional error of at most $\eta n$.

\begin{lemma}\label{lem:optimization}
There is an algorithm \emph{\textsc{MinimizeSurrogate}} that,
given $C,U,t,\eta,p,s$ as above, returns a rational number
$\widehat F$ satisfying
\begin{equation*}
 F_C(t)\le\widehat F\le F_C(t)+2\eta n.
\end{equation*}
The running time, including the construction of $U$, is
$n^{O(m)}\operatorname{poly}(p,m)$. The output $\widehat F$ has bit length $\ell(\widehat F)=\mathrm{poly}(n,p,m)$.
\end{lemma}

We prove \cref{lem:optimization} in \cref{app:numerical-minimization}.

\subsection{Proof of the main theorem}

\begin{proof}[Proof of \cref{thm:main}]
If $n\le1$ or the support graph has no perfect matching,
the algorithm returns $\per A$.
We may therefore assume $n\ge2$ and $\per A>0$.
Deleting edges that belong to no perfect matching preserves $\per A$.

By \cref{thm:entropy,lem:optimization},
\[
 \per C
 \le e^{F_C(t)/t}
 \le e^{\widehat F/t}
 \le e^{\alpha n}\per C,
\]
where the last inequality follows from
\begin{align*}
    \widehat F/t\leq F_C(t)/t+2\eta n/t\leq \log\per C+\frac{2nh(t)+2\eta n}{t}
\end{align*}
and the fact that
$(2h(t)+2\eta)/t\le\alpha$. Indeed, $s\ge2$ and
\[
 \log(e/\delta)=1+(s+4)\log2+\log(s+1)\le3(s+1),
\]
so $h(t)=h(\delta)\le\delta\log(e/\delta)\le3\alpha/16$.
Together with $\eta=\alpha/16$ and $t\ge1/2$, this gives the claim.

Since $1\le\per C\le e^{\widehat F/t}$, rounding
$e^{\widehat F/t}$ upward by at most $\eta n^{-n}$ in \cref{alg:permanent} gives
\[
 e^{\widehat F/t}\le Q_C
 \le e^{\widehat F/t}+\eta n^{-n}
 \le(1+\eta)e^{\widehat F/t}.
\]
Combining this with \cref{lem:scale-and-round}, the returned value
$Q=Q_C/S$ satisfies
\[
 \per A
 \le Q
 \le(1+\alpha)(1+\eta)e^{\alpha n}\per A
 \le e^{2\alpha n}\per A
 \le(1+\varepsilon)^n\per A.
\]
The last two inequalities use $\eta\le\alpha$, $n\ge2$, which gives
\begin{align*}
(1+\alpha)(1+\eta)e^{\alpha n}\leq e^{\alpha}e^{\eta}e^{\alpha n}\leq e^{\alpha n+2\alpha}\leq e^{2\alpha n},
\end{align*}
and
$2\alpha\le\varepsilon/2\le\log(1+\varepsilon)$.

\smallskip
\noindent\emph{Running time.}
We first test whether the support graph has a perfect matching.
An edge $uv$ belongs to a perfect matching precisely when deleting
$u$ and $v$ leaves a graph with a perfect matching.
Thus the preprocessing uses $O(n^2)$ bipartite perfect matching tests,
followed by \textsc{ScaleAndRound}, and takes
$\operatorname{poly}(n,s)$ time.
By \cref{lem:optimization}, constructing $U$ and computing
$\widehat F$ take $n^{O(m)}\operatorname{poly}(p,m)$ time, where
$p=O(n\log n+\log(2/\varepsilon))$ and
\begin{equation*}
 m=O\!\left(K\log\frac{2K}{\eta}\right)
  =O\!\left(\varepsilon^{-2}\log^3\frac2\varepsilon\right).
\end{equation*}
Since also $|\widehat F/t|=O(n\log n)$ by the bounds above,
evaluating $e^{\widehat F/t}$ to additive accuracy
$\eta n^{-n}$ takes $\operatorname{poly}(n,p,m)$ time,
by \cref{lem:exponential-enclosures}.
This gives the claimed running time.

For rational input, the implementations in
\cref{lem:scale-and-round,lem:optimization,app:numerical-minimization}
and the final division by $S$ also give the claimed
bound on intermediate bit length.
\end{proof}

\bibliographystyle{amsplain0}
\bibliography{main}

\appendix
\section{Numerical optimization}\label{app:numerical-minimization}

We prove \cref{lem:optimization} using gradient steps followed by
projection onto a box. At each step, we use the gradient of $\log d_v$ if the estimated
degree of some vertex $v$ is too large, and the gradient of $f$
otherwise.
This is an inexact version of the constrained subgradient method of
\cite[Sections~6--7]{BoydSubgradient}; we give the argument below,
including the errors from numerical evaluation.

Let $\theta^*$ be the minimizer from \cref{lem:field-bound}, and put
$X=[-R_0,R_0]^{2n}$. Since $K=4\delta^{-2}$, the bounds in
\cref{eq:degree-bound,eq:field-radius,eq:surrogate-minimum}  give
\[
 \theta^*\in X,\qquad
 d_v(\theta^*)\le\frac{t}{\delta^2}<\frac K4,
 \qquad
 f(\theta^*)\le F_C(t)+\eta n.
\]
Recall that we are given the polynomial $U$ as in \cref{thm:surrogate}, with degree  $m=O(K\log\frac{2K}{\eta})$. 
Let $\Pi_X:\R^{2n}\to X$ denote coordinatewise clipping,
$\Pi_X(\theta)_v=\min\{R_0,\max\{-R_0,\theta_v\}\}$.

The following procedure returns a number $\widehat F$ such that $F_C(t)\leq \widehat F\leq F_C(t)+2\eta n$. Recall the choice of parameters in \cref{eq:algorithm-parameters,eq:field-radius}.

\begin{algorithm}[htbp]
\caption{\textsc{MinimizeSurrogate}}
\label{alg:minimize-surrogate}
\begin{algorithmic}[1]
\Require $C,U,t,\eta,p,s$
\Ensure An approximation $\widehat F$ to $F_C(t)$

\State Set $h=\frac{\eta}{64n^2}$ and  $T=1+\left\lceil\frac{4nR_0^2}{h\eta}\right\rceil$.
\State Set $\theta_0=0$ and $\mathcal I=\emptyset$.

\For{$k=0,1,\dots,T-1$}
    \State For each $v\in V$, compute an estimate
    $\widehat d_v(\theta_k)$ satisfying
    $|\widehat d_v(\theta_k)-d_v(\theta_k)|\le \frac K{16}$. \label{line:dv}

    \If{there exists $v$ such that
    $\widehat d_v(\theta_k)>3K/8$} \label{line:case-a}
        \State Pick any such $v$ and set
        $g_k=\nabla\log d_v(\theta_k)$.
    \Else \label{line:case-b}
        \State Set
        $g_k=\nabla f(\theta_k)$,
        and add $k$ to $\mathcal I$. 
        \State Compute and record an  approximation
        $\widehat f(\theta_k)$ satisfying
        $f(\theta_k)
        \le \widehat f(\theta_k)
        \le f(\theta_k)+\frac{\eta}{2}$. \label{line:tildef}
    \EndIf
    \State Compute a dyadic vector $\widetilde g_k$ such that $\|\widetilde g_k-g_k\|_2\leq\frac{\eta}{16nR_0}$. \label{line:gk}
    \State Set $\theta_{k+1}=\Pi_X(\theta_k-h\widetilde g_k)$. \label{line:thetak+1}
\EndFor
\State Return  $\widehat F=\min_{k\in\mathcal I}\widehat f(\theta_k)$\label{line:final}
\end{algorithmic}
\end{algorithm}

\medskip

We remark that in the above procedure, we never compute $g_k$ exactly; instead, we always use the dyadic approximation $\widetilde g_k$.

\newcommand{\lineref}[1]{Line~\ref{#1}}

\subsection{Correctness} In this subsection, we show that \cref{alg:minimize-surrogate} returns the desired $\widehat F$. Observe that, for each round $k$:

\begin{itemize}
    \item If we enter \lineref{line:case-a}, then for the chosen $v$ with $\widehat d_v(\theta_k)>3K/8$, we have
    $d_v(\theta_k)>3K/8-K/16=5K/16>K/4>d_v(\theta^*)$.
    In particular, $\frac{d_v(\theta_k)}{d_v(\theta^*)}>5/4$.
    \item If we enter \lineref{line:case-b}, then
all vertices $v$ have $\widehat d_v(\theta_k)\leq 3K/8$, and thus $ d_v(\theta_k)\leq 3K/8+K/16<K/2$.
    In particular, since the projection \lineref{line:thetak+1} ensures that $\|\theta_k\|_\infty\leq R_0$ for all $k$, we get that $\theta_k\in \mathcal D$.
\end{itemize}

\begin{lemma}\label{lem:gk2}
For all $k$, we have
    $\|g_k\|_2\leq 2\sqrt{2n}$.
\end{lemma}
\begin{proof}
For each edge $e=uv$, let $b_e=\mathbf 1_{\{u,v\}}\in\R^{2n}$.
If we enter \lineref{line:case-a},
\[
 \nabla\log d_v(\theta_k)
 =\sum_{e\ni v}\frac{w_e(\theta_k)}{d_v(\theta_k)}b_e
\]
is a convex combination of vectors of norm $\sqrt2$, so
$\|g_k\|_2\le\sqrt2$. If we enter \lineref{line:case-b}, $\theta_k\in\mathcal D$. By
\cref{eq:vertex-marginal,eq:taylor-gradient-error,eq:surrogate-degree},
using $0<t<1$ and
$\partial_{\theta_v}W(w_C(\theta))=d_v(\theta)$, we have
\[
 |\partial_{\theta_v}f(\theta_k)|
 \le1+2\beta_md_v(\theta_k)
 \le1+\eta/2\le2.
\]
Thus $\|g_k\|_2\le2\sqrt{2n}$ in either case.
\end{proof}

\begin{lemma}\label{obs:decrease}
Fix $k$. Suppose in round $k$ we enter \lineref{line:case-a}, or we enter \lineref{line:case-b} with $f(\theta_k)>f(\theta^*)+\eta/2$. Then we have $\|\theta_{k+1}-\theta^*\|_2^2\leq \|\theta_{k}-\theta^*\|_2^2-\frac{h\eta}{2}$.
\end{lemma}
\begin{proof}
If we enter \lineref{line:case-a},  convexity of $ \log d_v$ gives
$\langle \nabla \log d_v(\theta_k),\theta_k-\theta^*\rangle\geq \log d_v(\theta_k)-\log d_v(\theta^*)>\log(5/4)>\eta/2$;
if we enter \lineref{line:case-b}, convexity of $f$ on $\mathcal D$, which contains
$\theta_k$ and $\theta^*$, gives the lower bound
$\langle \nabla f(\theta_k),\theta_k-\theta^*\rangle\geq f(\theta_k)-f(\theta^*)>\eta/2$. In either case, we have $\langle  g_k,\theta_k-\theta^*\rangle>\eta/2$.
Since $\|\theta_k-\theta^*\|_2\le2\sqrt{2n}R_0\le2nR_0$,
\[
 |\langle\widetilde g_k-g_k,\theta_k-\theta^*\rangle|
 \le\frac{\eta}{16nR_0}\,2nR_0=\eta/8.
\]
Hence $\langle\widetilde g_k,\theta_k-\theta^*\rangle>3\eta/8$.
The preceding \cref{lem:gk2} also gives
$\|\widetilde g_k\|_2\le2\sqrt{2n}+\eta/(16nR_0)\le4n$.

Projection onto $X$ cannot increase distance from $\theta^*\in X$.
The standard projected-step inequality
\cite[Section~6]{BoydSubgradient} therefore gives
\[
\begin{aligned}
 \|\theta_{k+1}-\theta^*\|_2^2
 &\le \|\theta_k-h\widetilde g_k-\theta^*\|_2^2= \|\theta_k-\theta^*\|_2^2
   -2h\langle\widetilde g_k,\theta_k-\theta^*\rangle
   +h^2\|\widetilde g_k\|_2^2\\
 &\le\|\theta_k-\theta^*\|_2^2-3h\eta/4+16n^2h^2
  =\|\theta_k-\theta^*\|_2^2-h\eta/2,
\end{aligned}
\]
where the last equality uses $h=\eta/(64n^2)$.
\end{proof}

We can now show the correctness of \cref{alg:minimize-surrogate}.

\begin{proof}[Proof of correctness]
Since $\|\theta_0-\theta^*\|_2^2\le2nR_0^2$, the decrease in
\cref{obs:decrease} cannot occur in all $T$ iterations.
Thus some $k_1\in\mathcal I$ satisfies
$f(\theta_{k_1})\le f(\theta^*)+\eta/2$.
Every recorded value is at least $F_C(t)$ by
\cref{eq:surrogate-minimum}, since its iterate lies in $\mathcal D$.
Consequently,
\[
 F_C(t)\le\widehat F\le\widehat f(\theta_{k_1})
 \le f(\theta^*)+\eta
 \le F_C(t)+\eta n+\eta\le F_C(t)+2\eta n.
\]
\end{proof}

\subsection{Running time and bit length of output}

We now show that \cref{alg:minimize-surrogate} finishes in $n^{O(m)}\text{poly}(p,m)$ time, outputting $\widehat F$ with $\ell(\widehat F)=\text{poly}(n,p,m)$.
We will use the following two estimates.

\begin{lemma}\label{lem:exponential-enclosures}
Let $x$ be a rational number of bit length $\ell$, let $B\ge1$ be an
integer with $|x|\le B$, and let $Q\ge0$ be an integer. In
$\operatorname{poly}(B,Q,\ell)$ bit operations, one can compute a
nonnegative dyadic number $L$ such that $ L\le e^x\le L+2^{-Q}$.
\end{lemma}

\begin{proof}
Set $N=8B+Q+4$, $r=2^{-(Q+2)}$, and compute $ A=\sum_{j=0}^{N}\frac{x^j}{j!}$.
Taylor's remainder bound and $k!\ge(k/e)^k$ give
\[
 |e^x-A|
 \le \frac{e^B B^{N+1}}{(N+1)!}
 \le e^B\left(\frac{eB}{N+1}\right)^{N+1}
 \le 2^{2B-(N+1)}
 \le r,
\]
where we used $e<4$ and $N+1\ge8B$.
Put $a=\max\{0,A-r\}$ and round $a$ downward to the nearest
multiple $L$ of $2^{-(Q+1)}$. Then $0\le L\le e^x$, and
\[
 e^x-L\le 2r+2^{-(Q+1)}=2^{-Q}.
\]
There are $N+1=O(B+Q)$ Taylor terms. If $x=a_0/b_0$, their
denominators divide $b_0^N N!$, whose bit length is
$O(N\ell+N\log(N+1))$. Every partial sum has absolute value
at most $e^B$, so its numerator requires only $O(B)$
additional bits with this denominator. Thus the rational
arithmetic and final rounding take
$\operatorname{poly}(B,Q,\ell)$ bit operations.
\end{proof}

We now introduce the following notation on the polynomial $U$.

\begin{definition}
    Write $U$ as a sum of its distinct nonzero monomials,
\[
 U(w)=\sum_{j=1}^S u_jw^{\gamma_j}
     =\sum_{j=1}^S u_j\prod_{e\in E}w_e^{(\gamma_j)_e},
\]
where each $u_j\in\Q$ and
$\gamma_j\in\mathbb Z_{\ge0}^E$.
Substituting $w_C(\theta)_e=C_ee^{\theta_u+\theta_v}$ and setting
$a_j=u_j\prod_{e\in E}C_e^{(\gamma_j)_e}$ and
$(b_j)_v=\sum_{e\ni v}(\gamma_j)_e$, we get
\[
 f(\theta)=\sum_{j=1}^S a_je^{b_j\cdot\theta}-t\sum_v\theta_v,
 \qquad
 \frac{\partial f}{\partial\theta_v}(\theta)
 =\sum_{j=1}^S a_j(b_j)_ve^{b_j\cdot\theta}-t.
\]
\end{definition}

Furthermore, let $H_U=m(2(n^2+1))^m$ and $D_U=m!(4K)^m(4K+1)^m(8K+1)^m$. 

\begin{lemma}\label{obs:U-coeffs} 
We have
\[
 1\le S\le(n^2+1)^m,\qquad
 |a_j|\le H_U,\qquad
 0\le(b_j)_v\le m,\qquad
 |b_j\cdot\theta|\le2mR_0\quad(\theta\in X).
\]
The denominators of the coefficients $u_j$ divide $D_U=m!(4K)^m(4K+1)^m(8K+1)^m$,
and those of the coefficients $a_j$ divide $D_U2^{pm}$.
\end{lemma}

\begin{proof}
Put $d=|E|\le n^2$, and let $\|P\|_{\mathrm{coeff}}$ denote
the sum of the absolute values of the coefficients of a
polynomial $P$. Recall from the construction of $U$ that
\[
 s_\ell=\sum_{j=1}^{\min\{\ell,\lfloor N/2\rfloor\}}
 \rho^j\binom{\ell-1}{j-1}M_j,\qquad \rho=(4K)^{-1}.
\]
Since $M_j$ has at most $d^j$ monomials, each with coefficient
$1$, and $\rho\le1$, we have
$\|s_\ell\|_{\mathrm{coeff}}
 \le d(1+d)^{\ell-1}\le(1+d)^\ell$.
Expanding the formal logarithm gives
\[
 \lambda_\ell=
 \sum_{r=1}^{\ell}\frac{(-1)^{r+1}}r
 \sum_{\substack{\ell_1+\cdots+\ell_r=\ell\\\ell_i\ge1}}
 s_{\ell_1}\cdots s_{\ell_r}.
\]
There are $\binom{\ell-1}{r-1}$ terms in the inner sum.
Since the coefficient norm is submultiplicative,
\[
 \|\lambda_\ell\|_{\mathrm{coeff}}
 \le(1+d)^\ell
       \sum_{r=1}^{\ell}\frac1r\binom{\ell-1}{r-1}
 \le\frac{(2(1+d))^\ell}{2}.
\]
Now
$U=\sum_{\ell=1}^m z_\star^\ell\lambda_\ell+\beta_mW$,
where $z_\star<1$ and $\beta_m\le1$. Thus
\[
 \|U\|_{\mathrm{coeff}}
 \le\frac m2(2(n^2+1))^m+d
 \le m(2(n^2+1))^m=H_U.
\]

We have chosen $K$ to be an integer. The coefficient
denominators of $s_\ell$ divide $(4K)^\ell$, so the
formal logarithm formula shows that those of $\lambda_\ell$
divide $\ell!(4K)^\ell$. Since
\[
 z_\star=\frac{4K}{4K+1},\qquad
 \beta_m=\frac{400K(8K)^{m+1}}{(8K+1)^m},
\]
the coefficient denominators of $U$ divide $D_U$.
Consequently, every coefficient of $U$ has bit length
\[
 O\!\left(m\log(m+1)+m\log(K+1)+m\log(n+1)\right).
\]
Every $C_e$ is a multiple of $2^{-p}$, and each monomial
has degree at most $m$, giving the stated denominator
bound for $a_j$. Also, $|a_j|\le|u_j|\le H_U$ because
$0\le C_e\le1$.

There are at most $(d+1)^m$ monomials of degree at most $m$, as
each can be represented by $m$ factors chosen from
$\{1\}\cup\{w_e:e\in E\}$.
The graph has an edge, and the linear coefficient of each
$w_e$ in $U$ is
$\rho\sum_{\ell=1}^m z_\star^\ell+\beta_m>0$, so $S\ge1$.
Finally,
\[
 (b_j)_v\le\sum_e(\gamma_j)_e\le m,\qquad
 \sum_v(b_j)_v=2\sum_e(\gamma_j)_e\le2m.
\]
The last identity gives $|b_j\cdot\theta|\le2mR_0$ on $X$.
\end{proof}

We now introduce the following accuracy parameters that will be used in the execution of \cref{alg:minimize-surrogate}.

\begin{definition}\label{def:PQdQf}
    Let 
\[
\epsilon_0:=\frac{\eta}{64n^2R_0},\qquad \omega:=\frac{K\epsilon_0}{64n},\qquad \omega_2:=\frac{\epsilon_0}{2mSH_U}.
\]
Put
\[
 P=\left\lceil\log_2(1/\epsilon_0)\right\rceil,\qquad
 Q_d=\max\left\{0,\left\lceil\log_2(1/\omega)\right\rceil\right\},
 \qquad
 Q_f=\max\left\{0,\left\lceil\log_2(1/\omega_2)\right\rceil\right\}.
\]
In particular, we have $2^{-P}\leq\epsilon_0$,
$2^{-Q_d}\leq\omega$, and $2^{-Q_f}\leq\omega_2$. 
 
\end{definition}

Further,  let $Y:=\max_{k,v}\ell((\theta_k)_v)$ throughout \cref{alg:minimize-surrogate}. We will eventually give an upper bound on $Y$.

\begin{proof}[Proof of upper bounds on running time and $\ell(\widehat F)$] We know from \cref{thm:surrogate} that $U$ can be constructed in $n^{O(m)}$ time.
    We now describe more concretely how one can execute each computational step in \cref{alg:minimize-surrogate}.

Recall that for every $k=0,1,\dots,T-1$, we perform a subset of the following.
\begin{enumerate}[leftmargin=*]
    \item \lineref{line:dv}. In this step, we compute $\widehat d_v(\theta_k)$ satisfying
$|\widehat d_v(\theta_k)-d_v(\theta_k)|\le K/16$, as follows:
\begin{itemize}[leftmargin=*]
    \item For every $e=uv$, find a nonnegative dyadic number
$\widehat x_{uv}\in 2^{-Q_d-1}\mathbb Z$ with
$|\widehat x_{uv}-e^{(\theta_k)_u+(\theta_k)_v}|\le  2^{-Q_d}\leq \omega$.

\textit{Running time.} Since $\ell((\theta_k)_v)\leq Y$,
By \cref{lem:exponential-enclosures},  this takes $\text{poly}(Y,R_0,n,Q_d)$ time. Also, the resulting $\ell(\widehat x_{uv})=\text{poly}(Q_d,R_0)$.

\item For every $v$, take $ \widehat d_v(\theta_k)=\sum_{u\sim v}C_{uv}\widehat x_{uv}$. 

\textit{Running time.} Since $\ell(C_{uv})=\text{poly}(p)$ and $\ell(\widehat x_{uv})=\text{poly}(Q_d,R_0)$, this takes $\text{poly}(n,p,Q_d,R_0)$ time. Also, the resulting $\ell(\widehat d_v(\theta_k))=\text{poly}(\log n,p,Q_d,R_0)$. 
\end{itemize}

\noindent \textit{Correctness.} The above procedure gives $$|\widehat d_v(\theta_k)- d_v(\theta_k)|\leq \sum_{u\sim v}C_{uv}|\widehat x_{uv}-e^{(\theta_k)_u+(\theta_k)_v}|\leq  n\omega=\frac{K\epsilon_0}{64}\leq K/16.$$
    \item \lineref{line:gk} after \lineref{line:case-a}. In this step, we choose a dyadic  vector $\widetilde g_k$ such that $\norm{\widetilde g_k-\nabla\log d_v(\theta_k)}_2\le\frac{\eta}{16nR_0}$, as follows:
    \begin{itemize}[leftmargin=*]
        \item Compute $ \overline g_k=\sum_{u\sim v}
       \frac{C_{uv}\widehat x_{uv}}{\widehat d_v(\theta_k)}b_{uv}$,
where we recall that $b_{uv}=\mathbf 1_{\{u,v\}}$.

\textit{Running time.} Combining the bit length bounds in (1), we get that this takes $\text{poly}(n,p,Q_d,R_0)$ time. The resulting $\overline g_k$ has $\ell(\overline g_k(v))\leq \text{poly}(\log n, p,Q_d,R_0)$ pointwise.

\item Round $\overline g_k$ to the nearest dyadic vector $\widetilde g_k\in (2^{-P-1}\mathbb Z)^{2n}$, with $|\widetilde g_k(v)-\overline g_k(v)|\leq 2^{-P}\leq\epsilon_0$.

\textit{Running time.} This takes $\text{poly}(n,p,Q_d,R_0,P)$ time.

    \end{itemize}

\noindent \textit{Correctness.} 
Recall that
$d_v(\theta_k)>5K/16$ and
$\widehat d_v(\theta_k)\ge d_v(\theta_k)-n\omega>K/4$.
The coordinate at $v$ equals $1$ both in $\widetilde g_k$ and in $\nabla\log d_v(\theta_k)$, and the coordinates outside $\{v\}\cup\{u:u\sim v\}$ equal $0$ in both vectors.
For every  $u\sim v$, we have
\begin{align*}
 \left|\overline g_k(u)-
       (\nabla\log d_v(\theta_k))_u\right|
 &=
 \left|\frac{C_{uv}\widehat x_{uv}}{\widehat d_v(\theta_k)}
       -\frac{w_{uv}(\theta_k)}{d_v(\theta_k)}\right|\\
 &\le
 \frac{|C_{uv}\widehat x_{uv}-w_{uv}(\theta_k)|}
      {\widehat d_v(\theta_k)}
 +\frac{w_{uv}(\theta_k)}{d_v(\theta_k)}
  \frac{|d_v(\theta_k)-\widehat d_v(\theta_k)|}
       {\widehat d_v(\theta_k)}\\
 &\le\frac{(n+1)\omega}{d_v(\theta_k)-n\omega}
 \le\frac{(n+1)\epsilon_0}{16n}
 \le\frac{\epsilon_0}{8}.
\end{align*}
Hence $\|\overline g_k-\nabla\log d_v(\theta_k)\|_2\leq \sqrt{n}\cdot \frac{\epsilon_0}{8}$. Together with  $\|\widetilde g_k-\overline g_k\|_2\leq \sqrt{2n}\cdot\epsilon_0$, this gives $\|\widetilde g_k-\nabla\log d_v(\theta_k)\|_2\leq 2\sqrt{n}\epsilon_0\leq\frac{\eta}{16nR_0}$.

    \item \lineref{line:gk} after \lineref{line:case-b}. In this step, we choose a dyadic vector $\widetilde g_k$ such that $\norm{\widetilde g_k-\nabla f(\theta_k)}_2\le\frac{\eta}{16nR_0}$, as follows:
    \begin{itemize}[leftmargin=*]
        \item For each $1\le j\le S$, compute a dyadic number $E_j\in 2^{-Q_f-1}\mathbb Z$ with
$|E_j-e^{b_j\cdot\theta_k}|\le 2^{-Q_f}\leq \omega_2$.

\textit{Running time.} Since $0\leq (b_j)_v\leq m$ are integers, $\ell((\theta_k)_v)\leq Y$, and $|b_j\cdot \theta_k|\leq 2mR_0$, we know from \cref{lem:exponential-enclosures} that this takes $\text{poly}(S,m,R_0,Y,Q_f)$ time. Furthermore, we have $\ell(E_j)=\text{poly}(Q_f,m,R_0)$.

\item For every $v$, compute the rational number
 $\overline g_k (v)=\sum_{j=1}^S a_j(b_j)_vE_j-t$.

\textit{Running time.} Recall from \cref{obs:U-coeffs} that $\ell(a_j)\leq \text{poly}(\log H_U,\log D_U,p,m)$ and  $\ell((b_j)_v) = O(\log(m+1))$. Also recall that $\ell(t)=O(s)$. Combining with the previous bound on $\ell(E_j)$, we know that this takes $$\text{poly}(S,\log H_U,\log D_U,p,m,Q_f,R_0,s)$$ time. 
 Also, the resulting $\overline g_k$  has $\ell(\overline g_k(v))=\text{poly}(\log S,\log H_U,\log D_U,p,m,Q_f,R_0,s)$ pointwise.

\item As before, round $\overline g_k$ to the nearest dyadic vector $\widetilde g_k\in (2^{-P-1}\mathbb Z)^{2n}$, with $|\widetilde g_k(v)-\overline g_k(v)|\leq 2^{-P}\leq\epsilon_0$.

\textit{Running time.} This takes $\text{poly}(\log S,\log H_U,\log D_U,p,m,Q_f,R_0,s,P)$ time.
    \end{itemize}

\noindent \textit{Correctness.} Since $|a_j|\le H_U$,  $(b_j)_v\le m$, and $|E_j-e^{b_j\cdot\theta_k}|\leq\omega_2$, we have
\[
 \left|\overline g_k(v)-\frac{\partial f}{\partial\theta_v}(\theta_k)\right|
 \le mSH_U\omega_2=\epsilon_0/2,
\]
which gives $\|\overline g_k-\nabla f(\theta_k)\|_2\leq \sqrt{2n}\cdot \epsilon_0/2$. Together with  $\|\widetilde g_k-\overline g_k\|_2\leq \sqrt{2n}\cdot\epsilon_0$, this gives $\|\widetilde g_k-\nabla f(\theta_k)\|_2\leq 4\sqrt{n}\epsilon_0\leq\frac{\eta}{16nR_0}$.

    \item \lineref{line:tildef} (after \lineref{line:case-b}). In this step,  we compute an   approximation 
$\widehat f(\theta_k)$  such that $f(\theta_k)\leq \widehat f(\theta_k)\leq f(\theta_k)+\eta/2$, as follows:
\begin{itemize}[leftmargin=*]
    \item With $E_j$ as in (3), let
    \[
 \widetilde f(\theta_k)=\sum_{j=1}^S a_jE_j-t\sum_v(\theta_k)_v,
 \qquad
 \widehat f(\theta_k)=\widetilde f(\theta_k)+\eta/4.
\]

\textit{Running time.} Since  $\ell(\eta)=O(s)$, combining with the previous bounds, this takes $$\text{poly}(S,\log H_U,\log D_U,p,m,Q_f,R_0,s,n,Y)$$ time. Also, the resulting $\widehat f(\theta_k)$ has bit length $\text{poly}(\log S,\log H_U,\log D_U,p,m,Q_f,R_0,s,\log n,Y)$. This follows from fact that (i) the denominators of every $a_jE_j$ divides $D_U2^{pm+Q_f+1}$, (ii)  $|\sum_{j=1}^Sa_jE_j|\leq SH_U(e^{2mR_0}+1)$, (iii) $\ell(\eta)=O(s)$, and (iv) each $\ell((\theta_k)_v)\leq Y$.

\end{itemize}

\noindent \textit{Correctness.} We have 
\begin{align*}
    |\widetilde f(\theta_k)-f(\theta_k)|\leq SH_U\omega_2=\frac{\epsilon_0}{2m}\leq \eta/4,
\end{align*}
and thus $f(\theta_k)\leq  \widehat  f(\theta_k)\leq f(\theta_k)+\eta/2$.

\item \lineref{line:thetak+1}. In this step, we take $\theta_{k+1}=\Pi_X(\theta_k-h\widetilde g_k)$. Recall that $R_0\in\mathbb Z$, $h=\frac{1}{2^{s+10}n^2}$, $\theta_0=0$, and all  $\widetilde g_k\in  (2^{-P-1}\mathbb Z)^{2n}$. Thus, we can deduce inductively that every coordinate of every $\theta_{k}$ has denominator bit length $O(s+\log n+P)$, and thus,
\[
Y=\max_{k,v}\ell((\theta_k)_v)=\text{poly}(\log R_0,P,s,\log n).
\]
And thus, \lineref{line:thetak+1} takes
\[
\text{poly}(n,s,S,\log H_U,\log D_U,p,m,Q_d,Q_f,R_0,P)
\]
time.
\end{enumerate}

We iterate the above loop for  $T$ times. After that, we go into the final step   \lineref{line:final}  that outputs $\min_{k\in\mathcal I}\widehat f(\theta_k)$. 
Since  each $\ell(\widehat f(\theta_k))=\text{poly}(\log S,\log H_U,\log D_U,p,m,Q_f,R_0,s,\log n,Y)$ and taking the minimum requires at most $T$ comparisons, we know that \lineref{line:final}
takes $$\text{poly}(S,\log H_U,\log D_U,p,m,Q_f,R_0,\log n,s,T,P)$$ time.

Finally, we  know   from \cref{eq:algorithm-parameters}, \cref{eq:field-radius},  \cref{thm:surrogate}, \cref{lem:scale-and-round}, \cref{obs:U-coeffs}, and \cref{def:PQdQf} that
\begin{align*}
&K=4\delta^{-2}\leq m,\qquad R_0=n^2(n+p)\sqrt{K}, \qquad \eta^{-1}=\frac{\sqrt{K}}{2(s+1)}\leq\frac{\sqrt{m}}{2},\qquad
T=1+\left\lceil\frac{256n^3R_0^2}{\eta^2}\right\rceil,\\
&s  = O(\log(\eta^{-1})),\qquad 
\log H_U  = O\bigl(\log m+m\log n\bigr),\qquad 
\log D_U  = O\bigl(m\log m+m\log K\bigr),\\
&\log S  = O(m\log n),\qquad
P  = O(\log n+\log R_0+\eta^{-1}),\\
&Q_d  = O(P+\log n),\qquad
Q_f  = O(P+\log m+\log S+\log H_U).
\end{align*}
This gives
\[
R_0,n,Q_d,p,m,Q_f,\log H_U,\log D_U,s,T,P=\text{poly}(n,p,m),\qquad   S= n^{O(m)}.
\]
Thus, the above analysis shows that \cref{alg:minimize-surrogate} finishes in $n^{O(m)}\text{poly}(p,m)$ time, with the final output $\widehat F=\min_{k\in\mathcal I}\widehat f(\theta_k)$ satisfying $\ell(\widehat F)=\text{poly}(n,p,m)$.
\end{proof}

\end{document}